\documentclass[final,1p,times,authoryear]{elsarticle}
\usepackage{amsmath}
\usepackage{amsthm}
\usepackage{amssymb}
\usepackage{amsfonts}

\usepackage{multirow}
\usepackage{multicol}
\usepackage{booktabs}

\usepackage{graphicx}
\usepackage{todonotes}
\usepackage{verbatim}

\usepackage{algorithm}
\usepackage{algpseudocode}

\usepackage{ulem}
\usepackage{xcolor}

\newcommand{\jac}{\textsf{jac}}

\newcommand{\B}{\mathcal{B}}
\newcommand{\FF}{\mathbb{F}}
\newcommand{\EE}{\mathbb{E}}
\newcommand{\rank}{\textsf{Rank}}

\newcommand{\x}{\textbf{x}}
\newcommand{\y}{\textbf{y}}
\newcommand{\wit}{\textsf{wit}}

\newtheorem{theorem}{Theorem}

\newtheorem{corollary}[theorem]{Corollary}
\newtheorem{conjecture}[theorem]{Conjecture}
\newtheorem{proposition}[theorem]{Proposition}
\newtheorem{definition}[theorem]{Definition}

\newtheorem{example}[theorem]{Example}
\newtheorem{remark}[theorem]{Remark}

\begin{document}

\begin{frontmatter}

\title{
On the Complexity of Solving Generic Over-determined Bilinear Systems
}

\author{John B. Baena}
\address{jbbaena@unal.edu.co }
\ead{jbbaena@unal.edu.co }

\author{Daniel Cabarcas}
\address{dcabarc@unal.edu.co }
\ead{dcabarc@unal.edu.co }

\author{Javier Verbel}
\address{javerbelh@unal.edu.co}
\ead{javerbelh@unal.edu.co}


\begin{abstract}
In this paper, we study the complexity of solving generic over-determined bilinear systems over a finite field $\FF$.
Given a generic bilinear sequence $B\in \FF[\x,\y]$, with respect to a partition of variables $\x$, $\y$, we show that, the solutions of the system $B = \textbf{0}$ can be efficiently found on the $\FF[\y]$-module generated by $B$.
Following this observation, we propose three variations of Gr\"{o}bner basis algorithms, that only involve multiplication by monomials in the $\y$-variables, namely, $\y$-XL, based on the XL algorithm, $\y$-MLX, based on the mutant XL algorithm, and $\y$-HXL, based on a hybrid approach.
We define notions of regularity for over-determined bilinear systems, that capture the idea of genericity, and we develop the necessary theoretical tools to estimate the complexity of the algorithms for such sequences.
We also present extensive experimental results, testing our conjecture, verifying our results, and comparing the complexity of the various methods.
\end{abstract}

\end{frontmatter}

\section{Introduction}


The problem of solving systems of polynomial equations has many important applications all over science and engineering.
The main abstraction to tackle the problem is the Gr\"{o}bner basis \citep{Buchberger_thesis}.
Due to its importance, the past three decades have seen qualitative improvements in the algorithms to solve the problem and in the understanding of its complexity \citep{Faugere, faugere2, XL_alg,bardet_theses,mxl1,MGB}.

The complexity of the problem is very sensitive to the structure of the system.
For generic systems (characterized as regular or semi-regular sequences) the complexity of the problem is well understood both classical \citep{bardet_theses}, and quantum complexity \citep{Grover-FHKKKP2017,Grover-BY18}.
There are also several works that adapt to very particular types of systems and analyze their complexity \citep{KS1999,sparsegroebnerunmixed,sp_minrank,bardet2020algebraic}.

In this paper, we study the complexity of solving determined and over-determined bilinear systems over a finite field.
More precisely, we are interested on systems of the form $B = \textbf{0}$, where $B$ is a sequence of $m$ bilinear polynomials in $n$ variables, with coefficients on a finite field $\FF$, and such that $n \leq m$.
By bilinear we mean that, there is a partition of the variables $\x = (x_1,x_2, \ldots , x_{n_x})$, $\y = (y_1,y_2, \ldots , y_{n_y})$ such that the quadratic part of every equation is some polynomial $f(\x,\y)$ such that for all $\lambda,\mu$, $f(\lambda\x,\mu\y)=\lambda\mu f(\x,\y)$.

\subsection{Our Contribution}

The key observation that drives this work is that, for a generic bilinear system $B = \textbf{0}$, its solutions can be efficiently found on the $\FF[\y]$-module generated by $B$, denoted by $I_{\y}(B)$.
This is in contrast to the typical approach of looking for a solution on the ideal generated by $B$.

Based on this observation, we propose variations of Gr\"{o}bner basis algorithms that only involve multiplication by mononials in the $\y$-variables.
We propose three such algorithms, $\y$-XL, based on the XL algorithm \citep{xl}, $\y$-MLX, based on the mutant XL algorithm \citep{DCabarcas}, and $\y$-HXL, based on the so called hybrid approach \citep{booleansolve}.

In order to analyze the complexity of these algorithms, we study the structure of $I_{\y}(B)$.
By looking at the Jacobian of $B$ with respect to $\x$, we show that in the determined and over-determined cases there are non-trivial syzygies of degree strictly less than $n_x$ in $\FF[\y]^{m}$.

We define a notion of $d$-regularity for homogeneous bilinear sequences by focusing on the $\FF[\y]$-module $I_{\y}(B)$.
In the same line, we define a notion of degree of regularity that captures the idea of the minimum degree at which the Hilbert polynomial equals the Hilbert function, but looking only at $I_{\y}(B)$ instead of at the ideal.
Supported on this notion of degree of regularity, we then define a notion of $\y$-semiregularity, applicable to determined or over-determiend bilinear sequences, as being $d$-regular for as high a degree $d$ as possible.
We compute this degree of such $\y$-semiregular sequences.
And we conjecture, based on extensive experimental evidence, that, for fixed parameters $n_x$, $n_y$, and $m$, there exists an open Zariski set $O$, contained in the set of all homogeneous bilinear sequences, such that every sequence in $O$ has this $\y$-semiregularity property.

The degree of regularity is an important value for measuring the complexity of Gr\"{o}bner basis algorithms, but it is not the only one.
Assuming a sequence $B$ is $\y$-semiregular, we also calculate the analogous of first fall degree
and witness degree in $I_{\y}(B)$. A subtle yet important contribution of this paper is a careful and clear explanation of each of these three degrees and the role they play in the complexity of different algorithms.

We then compute the complexity of the three proposed algorithms for $\y$-semiregular sequences.
We estimate that $\y$-XL solves the system in 
\begin{equation*}
 \mathcal{O}\left(m\binom{n_{y} +  \tilde{d} -2}{\tilde{d} -2} \left[n_x \binom{n_y + \tilde{d}-1 }{\tilde{d}-1}\right]^{\omega -1} \right)
 \end{equation*}
multiplications over $\FF_{q}$, where $2\leq \omega \leq 3$ is the linear algebra constant and  
\begin{equation*}
    \tilde{d} = \left\lceil \frac{n_y(n_x+1)}{m-n_x-1} \right\rceil + 1.
\end{equation*}
The complexity of $\y$-MXL is
 \begin{equation*}
    \mathcal{O}\left(m \binom{n_{y} + d_{\y -ff}(B)-2}{d_{\y -ff}(B)-2} \left[n_{x}\binom{n_{y} + d_{\y -ff}(B)-1}{d_{\y-ff}(B)-1}\right]^{\omega-1} \right)
\end{equation*}
multiplications over $\FF_{q}$, where  \begin{equation*}
    d_{\y - ff}(B) = \min \left\lbrace d\in \mathbb{Z} ^{+} \; \vert \;  d> \frac{n_{x}(n_{y}-1)}{m-n_{x}} + 1\right\rbrace .
\end{equation*}
For $a_x, a_y$ fixed, the complexity of $\y$-HLX$_{a_x,a_y}$ using Weidemann's algorithm is given by 
\begin{align*}
    \mathcal{O}\left(q^{a_x + a_y} (n_y - a_y+1)(n_x - a_x +1)^{3}\binom{n_y - a_y + \tilde{d} -1}{\tilde{d} -1}^{2} \right) 
\end{align*}
and using Gaussian elimination is
\begin{align*}
 \mathcal{O}\left(q^{a_x + a_y} m\binom{n_{y} -a_y +  \tilde{d} -2}{\tilde{d} -2} \left[(n_x -a_x +1) \binom{n_y -a_y  + \tilde{d}-1 }{\tilde{d}-1}\right]^{\omega -1} \right) 
\end{align*}
where $\tilde{d}$ is is given by
\begin{equation*}
\tilde{d} = \left\lceil \frac{(n_y-a_y)(n_x - a_x+1)}{m-n_x +a_x-1} \right\rceil +1.
\end{equation*}

We finally show extensive experimental evidence, testing our conjecture, verifying our results, and comparing the complexity to that of out-of-the-box algorithms.

\subsubsection{Related Work}

There are several works that have studied the solvability and complexity of the problem of finding solutions for bilinear systems of equations over any field $\mathbb{E}$.
 To the best of our knowledge, the first specialized methods for solving bilinear systems date from late 90's, with a work by \cite{CohenTomasi}. They studied the solvability of bilinear systems when $\mathbb{E}$ is the field of real numbers and proposed an algorithm to solve them

\cite{vinh2009solvability} studied the solvability of bilinear systems over finite fields and provided estimates for the number of solutions. The same year, Johnson and Link proposed an algorithm for solving bilinear systems over any field $\mathbb{E}$ \citep{JohnsonLink2009}. This is a deterministic and very efficient algorithm when $m = n_x n_y$.
However, the algorithm is probabilistic and, according to  our judgment, not efficient for $m < n_x n_y$. Based on the ideas of Johnson and Link,
\cite{DianYang_thesis} proposed an algorithm for solving bilinear systems over any field $\EE$, with $m< n_x n_y$. For this goal, a generic MinRank problem with $n_x n_y -m +1$ matrices over $\EE$ with target rank 1 needs to be solved  \citep[Sec.~2.6]{DianYang_thesis}.
They do not provide complexity estimates.

The complexity of solving a bilinear system over a finite field $\FF$ via  Gr\"{o}bner basis algorithms is analyzed in \citep{bihonogeneous2011}. They use the F5 algorithm, extending the F5 Criterion to avoid reductions to zero during the Gr\"{o}bner basis computation for bilinear ideals.
The extended criterion is named BILINF5CRITERION, and it works for the under-determined and determined cases, i.e, $m\leq n_x + n_y$ \citep[Prop.~1]{bihonogeneous2011}. They also provided an upper bound for the degree of regularity, which is used to estimate the complexity of computing a Gr\"{o}bner basis for the zero-dimensional and determined bilinear systems. Their estimate is
\begin{equation*}
    \mathcal{O}\left(\binom{n_x + n_y + \min\{n_x +1 ,n_y +1 \}}{\ min\{n_x +1 ,n_y +1 \}} ^{\omega }\right).
\end{equation*}
There are two main differences between the approach followed by Faugère et al. and the one presented in our work.

The first one is that they analyze the behavior of computing a Gr\"{o}bner basis algorithm for the ideal $I \subset \FF[\x,\y]$ generated by bilinear equations, while in this paper we analyze the same behavior but for the $\FF[\y]$-module, $I_{\y}$, generated by bilinear equations. The second difference is that in \citep{bihonogeneous2011} the complexity estimates are only meaningful when $m = n_x + n_y$ but not for the case $n _{x} + n_y \leq m$ as in this paper.

\cite{sparsegroebnerunmixed} considered a Gr\"{o}bner basis algorithm that does not use all monomials from a polynomial ring.
This method is applicable when the initially given polynomials are sparse and with the same support.
They use only monomials that appear in the support of the initial polynomials. For bilinear systems in $\FF[\x,\y]$ that means multiplying by monomials of degree $2d$ formed by a monomial of degree $d$ in $\FF[\x]$ and a monomial of degree $d$ in $\FF[\y]$.
This leads to a completely different approach to the one discussed in the present paper.

In cryptography, the security of several schemes can be broken via solving a system of bilinear equations \citep{KS1999,ZHFE_minrank, KStoHFE-, bardet2020algebraic}.
The complexity of solving such systems has been studied in \citep{bihonogeneous2011,sp_minrank,bardet2020algebraic}.
\cite{sp_minrank} and \cite{bardet2020algebraic} proposed two different modified Gr\"{o}bner basis algorithms for solving these particular bilinear systems, which only involve one group of variables during the Gr\"{o}bner basis computation. The complexity of these algorithms relies on the structure of the $\FF[\y]$-module $I_{y}$, which in these cases is generated by very particular bilinear equations. In this paper, we consider generic bilinear equations. We do not expect that the results presented here provided a tight upper bound for the bilinear systems considered in \citep{sp_minrank,bardet2020algebraic}. Instead, the present work generalizes the ideas of these papers.
\section{Preliminaries}

\subsection{Notation and Basic Definitions}
Throughout this paper, we adopt the following notation.

\begin{itemize}
\item $\FF $ denotes the field with $q$ elements.

\item  $\FF^{a\times b}$ denotes the ring of matrices of size $a\times b$ with entries in $\FF$. We use bold uppercase letters to denote matrices. Similarly, $\FF^c$ denotes the space of all vectors of length $c$ with entries in $\FF$, and we use bold lowercase letters to denote vectors. The entry of a matrix $\mathbf{A}$ indexed by $(i,j)$ is denoted by $\mathbf{A}[i,j]$.

\item We distinguish two sets of variables, the $\x$-variables and the $\y$-variables, represented respectively by the tuples $\x = (x_1.\ldots, x_{n_x} )$ and $\y = (y_1,\ldots , y_{n_y})$, with $n_{x} \leq n_{y}$.

\item The polynomial ring in $\x$-variables and $\y$-variables over $\FF$ is denoted by $\FF[\x,\y]$. It is doted with graded lexicographic monomial order, where $x_1>\cdots > x_{n_x} > y_1>\cdots > y_{n_y}$. $\FF[\x]$ (\textit{resp.} $\FF[\y]$) denote the subring of polynomials over $\FF$ in the $\x$-variables (\textit{resp.} $\y$-variables).

\item The \textit{degree of a sequence} of polynomials $S\in\FF[\x]^m$ is the maximum degree of the polynomials in $S$.

 
 \item $\FF[\x,\y]_{\alpha,\beta}$ denotes the set of homogeneous polynomials in $\FF[\x,\y]$ of degree $\alpha + \beta$, such that each of their monomials has degree $\alpha$ in the $\x$-variables and degree $\beta$ in $\y$-variables.
 
 \item  Let $S$ be a sequence of polynomials in $\FF[\x,\y]^{m}$. We say $S$ is 
 \begin{enumerate}
     \item  \textit{An under-determined sequence: } if $m < n_x + n_y$.

     \item \textit{A determined sequence: } if $m = n_x + n_y$.
     
     \item \textit{An over-determined sequence: } if $m > n_x + n_y$.
  \end{enumerate}
\end{itemize}

\begin{definition}[Jacobian]
Let $S =(f_1,\ldots, f_m)$ be a sequence of polynomials in $\FF[\x,\y]$. The Jacobian of $S$ with respect to $\x$ is the matrix $\jac_{\x}(S) \in \FF[\x,\y]^{m \times n_x}$ defined by  
\begin{displaymath}
\jac_{\x}(S)[i,j] = \frac{\partial f_i }{\partial x_j}. 
\end{displaymath}
The Jacobian $\jac_{\y}(S)$ is defined analogously.
\end{definition}

\begin{definition}[Zariski Topology]
The Zariski topology on $\FF^{k}$ is the topology whose closed sets are the algebraic subsets of $\FF^{k}$, i.e., all sets of the form 
$\left\lbrace \textbf{a} \in \FF^{k}\; \vert \; g(\textbf{a}) = 0, \; \forall g \in S \right\rbrace $ for some $S \subset \FF[x_1,\ldots, x_k]$.
\end{definition}

An open Zariski set is the complement of an algebraic set. Each open set is expected to be very large compared to $\FF^{k}$ because it is dense in $\FF^{k}$ \citep{Hartshorne1977}.

\subsection{Solving Overdetermined Quadratic Systems over Finite Fields}\label{sec:poly_solving}


The problem of solving a system of polynomial equations over a size $q$ field $\FF$ is closely related to the problem of finding a \textit{Groebner Basis} for an ideal in $\FF[\x]$. For instance, when a sequence of polynomials $S \in \FF[\x]^m$ has a unique solution over $\FF$, this solution can be found by finding a Groebner basis for the ideal generated by $S \cup \left\lbrace x_{i}^q - x_{i}\right\rbrace _{i = 1}^{n_x} $.

The most efficient Groebner basis algorithms (eg. XL \citep{xl}, Mutant XL \citep{DCabarcas} and F4 \citep{MPK:F:F4} follow an idea first explored by \cite{Lazard83}. They amass most of the computation on finding a staggered basis for the row space of the Macaulay matrix of $S \cup \left\lbrace x_{i}^q - x_{i}\right\rbrace _{i = 1}^{n_x} $.

\begin{definition}[Macaulay matrix]
The  Macaulay matrix of degree $d$ of a sequence of polynomials $S \in \FF[\x]^m$, denoted  $\textbf{M}_{\leq d}(S)$, is the matrix formed by the coefficients of all polynomials of the form $\mathfrak{m} f$, where $f \in S$ and $\mathfrak{m} \in \FF[\x]$ is a monomial of degree at most  $d-\deg(f)$. The columns of the matrix correspond to the monomials produced in all the products $\mathfrak{m} f$, and sorted in decreasing order with respect to the grevlex ordering. For $2 \leq j\leq d$, we use $\textbf{M}_{j}(S)$  to denote the row submatrix of $\textbf{M}_{\leq d}(S)$ formed by taking only monomials $\mathfrak{m} \in \FF[\x]$ with degree exactly $j-2$. We refer to $\textbf{M}_{j}(S)$ as the degree $j$ part of the Macaulay matrix $\textbf{M}_{\leq d}(S)$.
\end{definition}


The complexity of this kind of Gr\"{o}bner basis algorithms depend on the degree up to which the Macaulay matrix must be constructed, often called the \textit{solving degree}.
In certain cases, the solving degree can be approximated by
values that do not depend on the algorithm used, but only on the ideal $\langle S \rangle$ itself, for instance, the \textit{degree of regularity}, and the \textit{first fall degree}.

\begin{definition}[Degree of regularity]
 Let $S$ be a sequence of homogeneous polynomials over $\FF[\x]^m$. The degree of regularity of $S$ is defined as
 \begin{displaymath}
     d_{reg}(S) = \min \left\lbrace d\in \mathbb{Z} ^{+} \;\vert \; \dim\left(I_{d} \right) = \dim \left(\FF[\x]_{d}\right) \right\rbrace,
 \end{displaymath}
 where $I_{d}$ and $\FF[\x]_{d}$ are the $\FF$-vector spaces consisting of degree $d$ polynomials in $\langle S\rangle $ and $\FF[\x]$, respectively. If $S$ is not homogeneous, and $\tilde{S}$ is the homogeneous part of $S$ of highest degree, then $d_{reg}(S) := d_{reg}(\tilde{S})$.
\end{definition}

An equivalent way of defining $d_{reg}(S)$ is as the minimum integer $d$ such that the dimension of $I_{d}$ is $\binom{n_x  +d -1}{d}$. 
Syzygies and trivial syzygies are crucial concepts in understanding the complexity of Groebner basis algorithms.

\begin{definition}[Syzygy]
A syzygy of a sequence of polynomials $(f_1,\ldots , f_m) \in \FF[\x]^{m}$ is another sequence of polynomials $(g_1, \ldots ,g_{m}) \in F[\x]^{m}$ satisfying $\sum _{i=1}^{m} g_i f_i = 0.$
\end{definition}

\begin{definition}[Trivial Syzygy]
Let $S = (f_1,\ldots, f_m) \in \FF[\x]^{m}$. For any $1\leq i < j \leq m$, let $T_{i,j} = f_{i}\textbf{e}_{j} - f_{j}\textbf{e}_{i} $,
where $\textbf{e}_{i}$ is the $i$-th canonical basis vector.
Any element in the $\FF[\x]$-module generated by the vectors $\left\lbrace T_{i,j} \; \vert\; 1\leq i< j\leq m \right\rbrace
$ is called a trivial syzygy.
\end{definition}

The degree of regularity can be determined for a large family of sequences called semiregular. These are sequences with no relations among their polynomials except the trivial ones.

\begin{definition}[semiregular] 
A sequence $S = (f_1,\ldots ,f_m)\in \FF[\x]^m$ is called $d$-regular if for each $g \in \FF[\x] $ and for all $1\leq i \leq m$, the facts $gf_i \in \langle f_1,\ldots, f_{i-1}\rangle$ and $\deg(g f_i) <d$ imply $f_i \in \langle f_1,\ldots, f_{i-1}\rangle$.
A \textit{semiregular} sequence $S$ is a sequence that is $d_{reg}(S)$-regular.
\end{definition}
When $S$ is a sequence of homogeneous quadratic polynomials, $S$ is semiregular, if and only if, all the syzygies of $S$ of degree less than $d_{reg}(S) -2$ are trivial \citep[Prop.~ 6]{Bardet_asymptoticbehaviour}.

\begin{definition}[First Fall Degree]
 Let $S = (f_{1},\ldots, f_{m})$ be a sequence of quadratic polynomials in $\FF[\x]$. Let $\tilde{S}$ be the sequence of polynomials formed by the homogeneous part of largest degree of each polynomial in $S$. The first fall degree of $S$, denote by $d_{ff}(S)$, is  the minimum integer $d$ such that $\tilde{S}$ has a non-trivial syzygy of degree $d-2$.    
\end{definition}


For most quadratic sequences $S \in \FF[\x]^m$ the complexity of computing a Groebner basis for the ideal $I = \left\langle S\right\rangle $ is given by  
\begin{equation*}
    \mathcal{O}\left( \binom{n_{x} + d}{d}^\omega \right),
\end{equation*}
with $d=d_{reg}(S) = d_{ff}(S)$.
This is because most quadratic sequences $S$ have semiregular quadratic part $\tilde{S}$. In this case, there are no non-trivial degree falls below $d=d_{reg}(S)$. Furthermore, $\textbf{M}_{d}(\tilde{S})$ spans $\FF[\x]_d$, hence $d_{reg}(S) = d_{ff}(S)$, and also, for $k>d$, the leading term of any polynomial in $I_{k}$ is divisible by some polynomial in the row space of $\textbf{M}_{d}(S)$. 


\subsection{Bilinear Sequences over Finite Fields}

\begin{definition}[Bilinear Polynomial] \label{def:bilinear_seq}
 A quadratic polynomial $f \in \FF[\x,\y]$ is called bilinear with respect to $\x$ and $\y$ if it can be written as
 $$f = \x \textbf{A} \y ^{\top} + \textbf{b}\x ^{\top} + \textbf{c}\y ^{\top} + d ,$$
 where $\textbf{A}\in \FF^{n_x \times n_y} $, $\textbf{b} \in \FF^{n_x}, \textbf{c} \in \FF ^{n_y}$ and $d \in \FF$.
We use $\B(n_{x},n_{y},m)$ to denote the set of all length $m$ bilinear sequences in $\FF[\x,\y]$, where there are $n_{x}$ $\x$-variables and $n_{y}$ $\y$-variables. The subset of $\B(n_{x},n_{y},m)$ consisting only of homogeneous sequences is denoted by $\B^{(h)}(n_{x},n_{y},m)$.
\end{definition}


We now define the notion of generic bilinear sequences, which captures the properties of a sequence that only depend on the sequence being bilinear, without considering the particular coefficients appearing in the sequence. For instance, a property over $\FF[\x]$ that does not depend on the particular coefficients used is: any two polynomials $f_1,f_2 \in \FF[\x]$ satisfy $f_1f_2 - f_2 f_1 = 0$. These kinds of properties are called \textit{generic bilinear properties}, meaning that they are satisfied by generic bilinear sequences.

\begin{definition}
Let $\mathfrak{a}$ denote the set of parameters
$$ \{\mathfrak{a}_{{i,j}}^{k} \; \vert \; 1\leq i\leq n_x, \; 1\leq j \leq n_y, \; 1\leq k \leq m\} \cup \{\mathfrak{a}_\ell^{k} \; \vert \; 0\leq \ell \leq n_x + n_y, \; 1\leq k \leq m \}, $$
and let $\FF(\mathfrak{a})$ denote the ring of fractions of the polynomial ring $\FF[\mathfrak{a}]$. A generic bilinear sequence $\B(\mathfrak{a}) \in \FF(\mathfrak{a})[\x,\y]^{m}$ is a sequence of polynomials $(f_1,\ldots ,f_m)$, where for each $1\leq k \leq m$
$$ f_k = \sum_{i,j =1}^{n_x,n_y} \mathfrak{a}_{i,j}^{k} x_{i}y_{j} + \sum  _{i =1}^{n_x } \mathfrak{a}_{i} ^{k} x_{i} + \sum _{j=1}^{n_y} \mathfrak{a}_{n_x +j}^{k}y_{j} + \mathfrak{a}_{0}^{k}.$$
We say $\B(\mathfrak{a})=(f_1,\ldots f_m)$ is a generic homogeneous bilinear sequence if
$$ f_k = \sum_{i,j =1}^{n_x,n_y} \mathfrak{a}_{i,j}^{k} x_{i}y_{j}.$$
\end{definition}


The following propositions, introduced by \cite{sp_minrank}, highlight two generic bilinear properties.


\begin{proposition}\label{prop:syzygies_jac}
Let $\B(\mathfrak{a}) = \left(f_{1},f_{2},\ldots,f_{m}\right)$ be a generic homogeneous bilinear sequence  in $\FF(\mathfrak{a})[\textbf{x},\textbf{y}]^m$. Suppose $\mathcal{G} = (g_{1},g_{2},\ldots , g_{m})$ is a sequence in $\FF[\textbf{y}]^m$, then,
$
    \sum _{i=1}^m g_{i}f_{i} = 0
$,
if and only if, $\mathcal{G}^{\top}$ belongs to the left-kernel of $\jac_{\textbf{x}}(\B(\mathfrak{a}))$.
\end{proposition}

It is important to notice that, if a syzygy contains variables of only one set ($\textbf{x}$ or $\textbf{y}$), then that syzygy cannot be trivial. 


\begin{proposition}\label{prop:trivial_syzygy}
Let $\B(\mathfrak{a})$ be a generic homogeneous bilinear sequence in $\FF(\mathfrak{a})[\x,\y]^m$. If a sequence $\mathcal{G}\in \FF[\y]^m$ is a syzygy of $\B(\mathfrak{a})$, then $\mathcal{G}$ is nontrivial.
\end{proposition}


\section{Algebraic Structure}
Here we analyse the algebraic structure of the $\FF[\y]$-module generated by a bilinear sequence $B \in \B(n_x,n_y,m)$ over a finite field $\FF$. This  $\FF[\y]$-module is denoted by $I_{\y}(B)$, and is defined as the set of the linear combinations of polynomials in $B$ and coefficients in $\FF[\y]$.
Our approach is based on the particular structure of bilinear sequences, so we analyse the module of syzygies of a generic homogeneous sequence.
Notice that the field equations related to $\FF[\x,\y]$ do not have this structure. Thus, the natural procedure of concatenating the field equations to the original sequence and then apply the well-known theory of sequences over algebraic closed field --as in \citep{booleansolve}-- is not considered in this case.

\subsection{Jacobian Syzygies of Generic Bilinear Sequences}\label{sec:jac_syzygies}


\cite{bihonogeneous2011} combined Proposition \ref{prop:syzygies_jac} and Cramer's rule to find nontrivial syzygies for a generic homogeneous sequence. When $m> n_x$ this method provides  $\binom{m}{n_{x}+1}$ syzygies of degree $n_x$ in $\FF[\y]^m$. In the under-determined case, $m < n_x + n_y$, it was conjectured by \cite{bihonogeneous2011} that, when applied to sequences in $\FF[\x,\y]^{m}$, those syzygies form a basis for the left kernel of $\jac _{\x}(B)$, for each sequence $B$ in an open Zariski set $O \subset \B^{(h)}(n_x,n_y,m)$. This conjecture does not hold for the determined ($m = n_x + n_y$)  and over-determined ($m> n_x + n_y$) cases.

In the determined and over-determined cases, the degree $n_x$ syzygies described in \citep{bihonogeneous2011} do exist, but they do not form a basis for the left kernel of $\jac_{\x}(B)$. In those cases, the left-kernel of the $\jac_{\x}(B)$ is expected to have elements of degree less than $n_x$.


Let $B(\mathfrak{a})$ be a generic homogeneous sequence in  $\FF(\mathfrak{a})[\x,\y]^{m}$. Then, the Jacobian $\jac_{\x}(B(\mathfrak{a}))$ is a matrix of size $m \times n_{x}$, where each entry is a generic homogeneous linear form in $\FF(\mathfrak{a})[\y]$. 
Let $A$ be the set resulting of multiplying each row of $\jac_{\x}(B(\mathfrak{a}))$ by each degree $d-2$ monomial in $\FF[\y]$. So, $A$ is a set of
$$m\binom{n_y + d-3}{d-2}$$
elements living in the $\FF(\mathfrak{a})$-vector space of sequences of size $n_x$ containing degree $d-1$ polynomials in $\FF(\mathfrak{a})[\y]$. The coefficients of a vanishing $\FF(\mathfrak{a})$-linear combination of the elements of $A$ can be used to build a degree $d-2$ element in the left-kernel of $\jac_{\x}(B(\mathfrak{a}))$, and consequently, a degree $d-2$ nontrivial syzygy of $B(\mathfrak{a})$ via propositions \ref{prop:syzygies_jac} and  \ref{prop:trivial_syzygy}.

Since $\FF(\mathfrak{a})$ is a fraction field, every nonzero algebraic expression in $\FF(\mathfrak{a})$ has an inverse. 
 There is a vanishing $\FF(\mathfrak{a})$-liner combination of the elements of $A$ whenever
\begin{equation}\label{eq:row_col}
    m\binom{n_y + d-3}{d-2} > n_x\binom{n_y + d-2}{d-1}.
\end{equation}
Notice, Inequality \eqref{eq:row_col} holds if and only if
$$d > \frac{n_x(n_y-1)}{m-n_x} +1.$$
Therefore, the first fall degree of a generic bilinear sequence $\tilde{B}(\mathfrak{a}) \in \FF(\mathfrak{a})[\x,\y]^{m}$ is upper bounded by
$$
\operatorname{min} \left\lbrace d \in \mathbb{Z} ^{+} \; \vert  \;  d > \frac{n_x(n_y-1)}{m-n_x} +1\right\rbrace.
$$

In addition, it is easy to see that 
\begin{equation}\label{eqn:det_overdet}
    n_{x} + n_{y} \leq m \Longleftrightarrow\frac{n_{x}(n_y -1)}{m-n_x } +1 < n_{x}+1. 
\end{equation}
 Indeed,
\begin{align*}
    n_{x} + n_{y} &\leq m \\
    n_{x} + n_{y} -1 &< m \\
    n_y - 1 &< m-n_x \\
    \frac{n_y -1}{m-n_x } &< 1\\
    \frac{n_{x}(n_y -1)}{m-n_x } +1 &< n_{x}+1.
\end{align*}
%
%
Consequently, in the determined and over-determined cases, there exists an integer $d \leq n_x+1$ satisfying inequality \eqref{eq:row_col}, hence there are nontrivial  syzygies of degree strictly less than $n_x$ in $\FF(\mathfrak{a})[\y]^{m}$.

\subsection{$\y$-Degree of regularity}\label{sec:deg_reg}

The concepts of  $\y$-degree of regularity and $\y$-semiregularity for a homogeneous bilinear sequence are introduced in this section. Based on theoretical arguments and a wide experimental verification, we conjecture that most of the sequences in a space $\B^{(h)}(n_x,n_y,m)$ are $\y$-semiregular. We deduce the $\y$-degree of regularity for $\y$-semiregular sequences, see Proposition \ref{prop:deg_reg}.   
\begin{definition}
The $\y$-Macaulay matrix of degree $d$ of a sequence of bilinear polynomials $B \in \B(n_{x},n_{y},m)$ is defined as the matrix $\textbf{M}_{\y,\leq d}(B)$ containing the coefficients of the polynomials of the form $\mathfrak{m} f$, where $f \in B$ and $\mathfrak{m} \in \FF[\y]$ is a monomial of degree at most $d-2$. The columns of the matrix correspond to the monomials produced in all the products $\mathfrak{m} f$, and sorted in decreasing order with respect to the grevlex ordering. For $2 \leq j\leq d$, we use $\textbf{M}_{\y, j}(B)$  to denote the row submatrix of $\textbf{M}_{\y,\leq d}(B)$ formed by taking only monomials $\mathfrak{m} \in \FF[\y]$ with degree exactly $j-2$. We refer to $\textbf{M}_{\y,j}(B)$ as the degree $j$ part of the  $\y$-Macaulay matrix $\textbf{M}_{\y, \leq d}(B)$.
\end{definition}

\begin{example} Consider the sequence $ B= (x_{1}y_{1} +  x_{1}y_{2} + x_{2}y_{1},\; x_{2}y_{1} + x_{2}y_{2}) \in \B(2,2,2)$. Then the degree $3$ $\y$-Macaulay matrix $\textbf{M}_{\y,\leq 3}(B) $ is given by
\begin{displaymath}
\begin{array}{c}
y_1 f_2  \\
y_2 f_2  \\
y_1 f_1  \\
y_2 f_1  \\
f_2 \\
f_1
\end{array}
\underbrace{\begin{pmatrix}
0 & 0 & 0 & 1 & 1 & 0 &     0 & 0 & 0 & 0 \\
0 & 0 & 0 & 0 & 1 & 1 &     0 & 0 & 0 & 0 \\
1 & 1 & 0 & 0 & 0 & 1 &     0 & 0 & 0 & 0 \\
0 & 1 & 1 & 0 & 1 & 0 &     0 & 0 & 0 & 0 \\
0 & 0 & 0 & 0 & 0 & 0 &     0 & 0 & 1 & 1 \\
0 & 0 & 0 & 0 & 0 & 0 &     1 & 1 & 1 & 0
\end{pmatrix}}_{\textbf{M}_{\y,\leq 3}(B) }
\begin{pmatrix}
x_{1}y_{1}^{2}\\
x_{1}y_{1}y_{2}\\
x_{1}y_{2}^{2}\\
x_{2}y_{1}^{2}\\
x_{2}y_{1}y_{2}\\
x_{2}y_{2}^{2}\\
x_{1}y_{1} \\
x_{1}y_{2} \\
x_{2}y_{1} \\
x_{2}y_{2} 
\end{pmatrix},
\end{displaymath}
where $f_1 = x_{1}y_{1} +  x_{1}y_{2} + x_{2}y_{1}$ and $f_2 = x_{2}y_{1} + x_{2}y_{2}$. Notice that $\textbf{M}_{\y,3}(B)$ is the submatrix consisting of the first four rows of $\textbf{M}_{\y,\leq 3}(B)$, which were constructed multiplying $f_1$ and $f_2$ by monomials of degree $1=3-2$.
\end{example}

Since every row of the $\y$-Macaulay matrix $\textbf{M}_{\y,\leq d}(B)$ represents one polynomial, we can define the $\FF$-vector space $J_{\y,j}(B)$ generated by the polynomials represented by the rows of $\textbf{M}_{\y,j}(B)$, for any $2\leq j \leq d$. This is introduced in the following definition.

\begin{definition}
Given a bilinear sequence $B = (f_{1},\ldots, f_{m}) \in \B(n_{x},n_{y},m)$,  we use the symbols $I_{\y,\leq d}(B) $ and $J_{\y,\leq d}(B)$ to denote the following $\FF$-vector spaces 
\begin{align*}
        I_{\y,\leq d}(B)  &= \left\lbrace h \; \vert \; h = \sum_{i=1}^{m} g_{i}f_{i}, g_{i} \in \FF[\y]  \text{ and } \deg\left(h\right) \leq d \right\rbrace  \\
       J_{\y,\leq d}(B) &= \left\lbrace h \; \vert \; h = \sum_{i=1}^{m} g_{i}f_{i}, g_{i} \in \FF[\y]  \text{ and } \deg\left(g_i\right) \leq d-2\right\rbrace .
\end{align*}
We use $I_{\y,j}(B)$ (\textit{resp.} $J_{\y,j}(B)$) to denote the elements in $I_{\y,\leq d}(B)$ (\textit{resp.} $J_{\y,\leq d}(B)$) of degree $j$ (\textit{resp.} for $g_i$'s having exactly degree $j-2$). We use $I_{\y}(B)$ to denote $\cup _{j=1}^{\infty } I_{\y,j}$.  
\end{definition}


\begin{remark}\label{rem:IvsJ}
In general, for any bilinear sequence $B$ we have $J_{\y,\leq d}(B) \subseteq I_{\y,\leq d}(B)$, but equality does not always hold.
\end{remark}

A homogeneous sequence is called $\y$-$d$-regular if it has no syzygyes over $\FF[\y]$ of degree at most $d$. More precisely, 

\begin{definition}
A homogeneous bilinear sequence $B \in \B^{(h)}(n_{x},n_{y},m)$ is said to be $\y$-$d$-regular if for each $j = 2,\ldots, d$,  
\begin{equation*}
    \rank\left( \textbf{M}_{\y,j}(B)\right) = m\binom{n_{y}+j -3}{j-2}.
\end{equation*}
\end{definition}


The $\y$-degree of regularity of a homogeneous bilinear sequence $B \in \B^{(h)}(n_{x},n_{y},m)$ is the minimum integer $d$ such that every degree $d$ monomial in $\FF[\x,\y]$, which is linear in the $\x$ variables, belongs to $J_{\y,d}(B)$. 

\begin{definition}\label{def:deg_reg}
The $\y$-degree of regularity $d_{\y,reg}(B)$ of a homogeneous bilinear sequence $B \in \B^{(h)}(n_{x},n_{y},m)$ is defined to be the minimum integer $d$ satisfying that 
$$\rank\left(\textbf{M}_{\y,d}(B)\right) = 
n_x\binom{n_{y} + d -2}{d-1}.$$
Alternatively, it can be defined as 
\begin{displaymath}
    d_{\y,reg}(B) = \min \left\lbrace d \; \vert \; \dim \left( J_{\y,d}(B)\right)  = \dim \left(\FF[\x,\y]_{1,d-1}\right) \right\rbrace .
\end{displaymath}
\end{definition}

\begin{definition}
A homogeneous bilinear sequence $B\in \B^{(h)}(n_{x},n_{y},m)$ is $\y$-semiregular if it is $\y$-$d$-regular for every $d$ less than $d_{\y,reg}(B)$. A sequence $B \in \B(n_x,n_y,m)$ is said $\y$-semiregular if its quadratic homogeneous part is $\y$-semiregular. 
\end{definition}

In this section, we show that the syzygies in $\FF[\y]^{m}$ of a $\y$-semiregular bilinear sequence are only the ones we know there exist for a generic bilinear sequence $\FF(\mathfrak{a})[\x,\y]^{m}$, see Section \ref{sec:jac_syzygies}. Hence, $\y$-semiregular sequences can be thought as generic sequences in $\B(n_x,n_y,m)$. As a consequence, if a sequence $B \in \B(n_x,n_y,m)$ is $\y$-semiregular, then its $\y$-degree of regularity must be less than or equal to $n_x+2$.


\begin{proposition}\label{prop:deg_reg}
Let $n_{x},n_y,m$ be positive integers with $n_{x} + n_{y}\leq m$. If $B \in \B^{(h)}(n_{x},n_{y},m)$ is $\y$-semiregular, then
\begin{equation*}
    d_{\y,reg}(B) = \left\lceil \frac{n_x(n_y-1)}{m-n_x}\right\rceil +1.
\end{equation*}
\end{proposition}
\begin{proof}
Let $n_x, n_y$ and $m$ be positive integers with $n_x + n_y \leq m$. Suppose $B \in \B^{(h)}(n_x,n_y,m)$ is a $\y$-semiregular sequence and, for simplicity, let us set $\tilde{d} = d _{\y,reg}(B)$. If
$$ \tilde{d}<  \frac{n_x(n_y-1)}{m-n_x} +1$$
then, as mentioned at the end of Section \ref{sec:jac_syzygies}, we had
\begin{equation*}
m \binom{n_{y} + \tilde{d} -3}{\tilde{d}-2} < n_{x}\binom{n_{y} + \tilde{d} -2}{\tilde{d}-1} = \rank \left(\textbf{M}_{\y, \tilde{d}}(B)\right) .
\end{equation*}
Consequently, the number of rows of $\textbf{M}_{\y, \tilde{d}}(B)$ would be less than its rank, which is a contradiction. Thus $n_x(n_y-1)/(m-n_x) +1 \leq \tilde{d}$.

Now, assume there is an integer $d$  satisfying 
$$\frac{n_x(n_y-1)}{m-n_x} +1 \leq d < \tilde{d}.$$ Hence
$$ \rank \left( \textbf{M}_{\y, d}(B)\right) < n_{x}\binom{n_{y} + d -2}{d-1} \leq 
m \binom{n_{y} + d -3}{d-2},$$
where the strict inequality is provided by the definition of the $\y$-degree of regularity (see Definition \ref{def:deg_reg}).
Thus, $B$ is not $\y$-$d$-regular, for $d < \tilde{d}$. That contradicts the fact that $B$ is $\y$-semiregular. Therefore, we must have
\begin{equation*}
    d_{\y,reg}(B) = \left\lceil \frac{n_x(n_y-1)}{m-n_x}\right\rceil +1.
\end{equation*}
\end{proof}

\begin{remark}
Notice that Equation \eqref{eqn:det_overdet} and Proposition \ref{prop:deg_reg} imply that the $\y$-degree of regularity of a $\y$-semiregular sequence $B$ is less than or equal to $n_{x}+1$. \end{remark}

Based on extensive experimental results, and following an analogous approach to the one used in \citep{bihonogeneous2011}, we now conjecture that being $\y$-semiregular is a generic property in the set of over-determined homogeneous bilinear sequences.

\begin{conjecture}\label{conj:generic_semiregular}
Suppose $n_{x},n_{y},m$ are positive integers with $n_{x} + n_{y} \leq m$. There exists an open Zariski set $ O \subset \B^{(h)}(n_{x},n_{y},m)$ such that every homogeneous sequence $B \in O$ is $\y$-semiregular.
\end{conjecture}
See more details in Section \ref{sec:experimental} and Table \ref{tab:conjecture}.
\medskip


\subsection{First Fall Degree}
\label{sec:firstfall}

In this section we introduce the $\y$-first fall degree for bilinear sequences and we estimate it for $\y$-semiregular sequences.

\begin{definition}[$\y$-First Fall Degree]
Let  $B$ be a bilinear sequence and $\tilde{B} = (\tilde{f}_{1},\ldots, \tilde{f}_{m})$ be the homogeneous sequence formed by the quadratic part of every polynomial in $B$. We say $B$ has a $\y$-degree fall at degree $d$, if there is a sequence of degree $d-2$ homogeneous polynomials $G \in \FF[\y]^{m}$ that is a non-trivial syzygy of $\tilde{B}$. 
The $\y$-first fall degree of $B$, denoted $d_{\y-ff}(B)$, is the smallest $d$ such that $B$ has a $\y$-degree fall at degree $d$. 
\end{definition}

In general, a degree fall for a sequence of polynomials $B$ over $\FF[\x,\y]$ is obtained from a non-trivial syzygy over $\FF[\x,\y]$  of the sequence $\tilde{B}$, which is the one formed by the homogeneous part of highest degree of each polynomial in $B$. For semiregular sequences we can precisely predict at what degree non-trivial syzygies first appear \citep{Bardet_asymptoticbehaviour,solvdeg_degreg}.
An analogous prediction can be done if we only consider coefficients in $\FF[\y]$ instead of the whole polynomial ring $\FF[\x,\y]$.

\begin{proposition}\label{prop:first_fall}
Let $n_{x},n_y,m$ be positive integers with $n_{x} + n_{y}\leq m$. Let $B \in \B(n_x,n_y,m)$ be a sequence such that its quadratic homogeneous part $\tilde{B} \in \B^{(h)}(n_{x},n_{y},m)$ is $\y$-semiregular. Thus,
\begin{equation*}  
    d_{\y-ff}(B) = \min \left\lbrace d \in \mathbb{Z} ^{+} \; \vert  \;  d > \frac{n_x(n_y-1)}{m-n_x} +1\right\rbrace .
\end{equation*}
\end{proposition}
\begin{proof}
Let $n_{x},n_y,m$ be positive integers with $n_{x} + n_{y}\leq m$. Let $B=(f_1,\dots,f_m) \in \B(n_x,n_y,m)$ be a $\y$-semiregular sequence with quadratic homogeneous part denoted by $\tilde{B}=(\tilde{f}_1,\dots,\tilde{f}_m) \in \B^{(h)}(n_{x},n_{y},m)$.
Set
\[
\tilde{d}=\min \left\lbrace d \in \mathbb{Z} ^{+} \; \vert  \;  d > \frac{n_x(n_y-1)}{m-n_x} +1\right\rbrace .
\]
By Proposition \ref{prop:deg_reg} $d_{\y,reg}(\tilde{B}) = \left\lceil \frac{n_x(n_y-1)}{m-n_x}\right\rceil +1$. So $d_{\y,reg}(\tilde{B}) = \tilde{d}$ or $d_{\y,reg}(\tilde{B}) = \tilde{d}-1$. In any case, if $j<\tilde{d}$ we have
\begin{equation*}
\rank \left( \textbf{M}_{\y, j}(\tilde{B})\right)
=
m \binom{n_{y} + j -3}{j-2}
\le
n_{x}\binom{n_{y} + j -2}{j-1}.
\end{equation*}
The case  $d_{\y,reg}(\tilde{B}) = \tilde{d}-1$ happens when $\frac{n_x(n_y-1)}{m-n_x}$ is an integer, and this is equal to what occurs when the matrix is square and invertible at degree $d_{\y,reg}(\tilde{B})$.
Then, for each $j<\tilde{d}$, the rows of $\textbf{M}_{\y, j}(B)$ are linearly independent, hence there are not degree $j-2$ polynomials $g_{1},\ldots ,g_{m} \in \FF[\y]$ such that $\sum_{i=1}^{m} g_{i}\tilde{f}_{i} = 0$  , which implies there is not $\y$-degree falls of $B$ up to degree $\tilde{d}-1$.
At degree $\tilde{d}$ 
\begin{equation*}
m \binom{n_{y} + \tilde{d} -3}{\tilde{d}-2}
>
n_{x}\binom{n_{y} + \tilde{d} -2}{\tilde{d}-1},
\end{equation*}
then the rows of $\textbf{M}_{\y, \tilde{d}}(B)$ are linearly dependent.
Hence there exist degree $\tilde{d}-2$ polynomials $G = (g_{1},\ldots ,g_{m}) \in \FF[\y]^{m}$ such that $\sum _{i=1}^{m}g_{i}\tilde{f}_{i} = 0$, and by Proposition \ref{prop:trivial_syzygy}, $G$ is a non-trivial syzygy of $\tilde{B}$, thus $\tilde{d}$ is the $\y$-first fall degree of $B$.
\end{proof}

\begin{remark}
Note that if $B$ is a $\y$-semiregular sequence, then $d_{\y-ff}(B) = d_{\y,reg}(\tilde{B})+1$ always when $m-n_x$ divides $n_x(n_y-1)$. Otherwise, $d_{\y-ff}(B) = d_{\y,reg}(\tilde{B})$.
\end{remark}



\subsection{The Witness Degree}\label{sec:witnessdegree}

In this section we define and estimate the $\y$-witness degree for a bilinear sequences.
Analogously to the witness degree definition in \citep{booleansolve}, the $\y$-witness degree for an bilinear sequence $B$ is defined as the minimum integer $d$ such that all the polynomials in $I_{\y, \leq d}(B)$ can be written as an $\FF$-linear combination of the polynomials represented by the rows of the $\y$-Macaulay matrix $\textbf{M}_{\y, \leq d}(B)$. More precisely:

\begin{definition}[$\y$-Witness Degree]
Suppose $\FF$ is a field with $q$ elements and $B = \left(f_{1}, \ldots , f_{m}\right)$ is a bilinear sequence in $\B(n_{x},n_{y},m)$. The $\y$-witness degree $d_{\y,\wit}(B)$ of $B$ is defined as
\begin{displaymath}
d_{\y,\wit}(B) = \min\left\lbrace d\in \mathbb{Z} ^{+} \; \vert \; I_{\y,\leq d}(B) = J_{\y,\leq d}(B) \right\rbrace.
\end{displaymath}
\end{definition}


The $\y$-witness degree of a bilinear sequence $B$ can be upper-bounded in most cases by the
$\tilde{\y}$-degree of regularity of its homogenization $B^{(h)}$ (see Theorem \ref{teo:witnessbound}), which is simply the homogeneous bilinear sequence containing the $\tilde{\y}$-homogenization of the polynomials in $B$, as explained in the following definition.

\begin{definition}[$\tilde{\y}$-homogenization]
Let  $f $ be a polynomial in the $\FF$-span of $\cup _{j=0}^{\infty}\FF[\x,\y]_{1,j} \cup \FF[\y]$. We define the  $\tilde{\y}$-homogenization of $f$ as the homogeneous polynomial $f^{(h)} \in  \FF[\tilde{\x},\tilde{\y}]$ given by
\begin{displaymath}
f^{(h)}(\tilde{\x},\tilde{\y}) = x_{0}y_{0}^{\deg (f)-1}f\left(\frac{x_1}{x_0},\ldots ,\frac{x_{n_x}}{x_0}, \frac{y_1}{y_0},\ldots , \frac{y_{n_y}}{y_0} \right),
\end{displaymath}
where $\tilde{\x} = (x_{0}, \x)$ and $\tilde{\y} = (y_{0}, \y)$ are sets of variables with sizes $n_x +1$ and $n_y+1$, respectively. Conversely, if $\tilde{f}$ is a homogeneous polynomial in $\FF[\tilde{\x},\tilde{\y}]$, with $\tilde{\x} = (x_{0}, \x)$ and $\tilde{\y} = (y_{0}, \y)$, we define its $\y$-\textit{dehomogenization} as the polynomial $\tilde{f}(1,\x,1,\y)$ in $\FF[\x,\y]$. For a sequence $B = \left(f_1,\ldots, f_m\right)$, where $f_1,\ldots, f_m\in \cup _{j=0}^{\infty}\FF[\x,\y]_{1,j} \cup \FF[\y]$, we define its $\tilde{\y}$-homogenization $B^{(h)}$ as the sequence $B^{(h)} =  \left(f_1^{(h)},\ldots, f_m^{(h)}\right)$. Finally, for a sequence $\tilde{B}$ of homogeneous polynomials, we define its $\y$-dehomogenization as the sequence $B$ in which the $i$-th component is the $\y$-dehomogenization of the $i$-th component of $\tilde{B}$.
\end{definition}


In particular, if  $f = \sum_{i=1}^{n_x}\sum _{j=1}^{n_y} a_{i,j}x_{i}y_{j} + \sum_{i=1}^{n_x} b_{i}x_{i} + \sum_{j=1}^{n_y} c_{j}y_{j} +d  \in \FF[\x,\y]$ is a bilinear polynomial, the $\tilde{\y}$-homogenization $f^{(h)}$ of $f$ is the homogeneous bilinear polynomial in the sets of variables $\tilde{\x} = (x_{0}, \x)$, $\tilde{\y} = (y_{0}, \y)$, given by
\begin{displaymath}
f^{(h)}(\tilde{\x},\tilde{\y}) = \sum _{i=0}^{n_x}\sum _{j =0}^{n_y} a_{i,j}x_{i}y_{j}, 
\end{displaymath}
where $a_{i,0} = b_{i}$ for $i>0$, $a_{0,j} = c_{j}$ for $j >0$ and $a_{0,0} = d$. 

\begin{theorem}\label{teo:witnessbound}
Let $n_{x},n_y,m$ be positive integers with $n_{x} + n_y\leq  m -2$. If Conjecture \ref{conj:generic_semiregular} is true, then there is an open Zariski set $O \subset \B(n_x,n_y,m)$ such that each $B \in O$ satisfies the following property: the system $B=\textbf{0}$ has no solution and $1$ belongs to the $\FF$-vector space generated by the polynomials representing the rows of the $\y$-Macaulay matrix $\textbf{M}_{\y,\leq \tilde{d}}(B)$, where 
\begin{equation*}
    \tilde{d} = \left\lceil \frac{n_y(n_x+1)}{m-n_x-1} \right\rceil + 1.
\end{equation*}
Moreover, $d_{wit}(B) \leq \tilde{d}$. 
\end{theorem}
\begin{proof}
Let $n_{x},n_y,m$ be positive integers with $n_{x} + n_y\leq  m -2$. Assuming veracity of Conjecture \ref{conj:generic_semiregular}, there is an open set $\tilde{O} \subset \B^{(h)}(n_x +1,n_y+1,m)$ such that any homogeneous sequence  $\tilde{B} \in \tilde{O}$ is $\tilde{\y}$-semiregular.
Define $O $ as the set  $\{ \tilde{B}(1,\x,1,\y) \; \vert \; \tilde{B} \in \tilde{O}\} \subset \B(n_x ,n_y,m)$ . We claim that the set $O$ is an open Zariski set. Indeed, each sequence  $\tilde{B}(x_0,\x,y_0,\y) \in \tilde{O}$ can be uniquely identified with a vector in $\FF^{m[(n_x + 1)(n_y +1)]}$, and the same vector identifies the sequence $\tilde{B}(1,\x,1,\y)$. So, $\tilde{O}$ is a Zariski open subset of $\B^{(h)}(n_x +1,n_y+1,m)$ if and only if $O$ is an open Zariski subset of $\B(n_x ,n_y,m)$. 

We will now show that each sequence in $O$ satisfies the property stated in the theorem. Recall that $\tilde{\x} = (x_{0},x_{1},\ldots , x_{n_{x}})$, $\tilde{\y} = (y_{0},y_{1},\ldots , y_{n_{y}})$ are the sets of variables and let $B$ be a bilinear sequence in $O$. Clearly, the $\tilde{\y}$-homogenization $B^{(h)}$ of $B$ is an element in $\tilde{O}$. Since  $B^{(h)}$ is $\tilde{\y}$-semiregular, any monomial of the form $f(\tilde{\x},\tilde{\y}) = x_{i}\tilde{\mathfrak{m}}$, where  $i =0, 1,2,\ldots, n_x$ and $\tilde{\mathfrak{m}} \in \FF[\tilde{\y}]$  is a monomial of degree $d_{\tilde{\y},reg}(B^{(h)})-1$, can be written as an $\FF[\tilde{\y}]$-linear combination of polynomials in $B^{(h)}$. That is, assuming 
$B^{(h)}(\tilde{\x},\tilde{\y}) = \left(f_1(\tilde{\x},\tilde{\y}),\ldots , f_m(\tilde{\x},\tilde{\y})\right) $, we have
\begin{equation*}
    f(\tilde{\x},\tilde{\y}) = \sum _{i=1}^{m} g_i(\tilde{\y})f_i(\tilde{\x},\tilde{\y}),
\end{equation*}
for some  $g_i(\tilde{\y})\in \FF[\tilde{\y}]$ having degree $d_{\tilde{\y},reg}(B^{(h)})-2$. Consequently, the polynomial $f(1,\x,1,\y)$, which is the $\y$-dehomogenization of $f(\tilde{\x},\tilde{\y})$, can be written as
\begin{equation}\label{eq:com}
    f(1,\x,1,\y) = \sum _{i=1}^{m} g_i(1,\y)f_i(1,\x,1,\y),
\end{equation}
where each $g_i(1,\y)$ is in  $\FF[\y]$ and has degree at most  $d_{\tilde{\y},reg}(B^{(h)})-2$. This means $f(1,\x,1,\y)$ belongs to $J_{\y, \leq \tilde{d}}(B)$, where $\tilde{d}= d_{\tilde{\y},reg}(B^{(h)})$. 

Notice that for every monomial $\mathfrak{m} \in \FF[\y]$ of degree at most $ d_{\tilde{\y},reg}(B^{(h)})-1$ and every variable $x_i $ in $\x$, there are monomials $h_1(\tilde{\x},\tilde{\y})=x_{0}\mathfrak{m}$ and $h_2(\tilde{\x}, \tilde{\y})=x_{i}\mathfrak{m}$ in $\FF[\tilde{\x},\tilde{\y}]$ of degree at most $d_{\tilde{\y},reg}(B^{(h)})$, such that $\mathfrak{m} = h_1(1,\x, 1,\y) $ and  $x_{i}\mathfrak{m} = h_2(1,\x, 1,\y) $.  Therefore, Equation \eqref{eq:com} implies that every monomial like $x_{i}\mathfrak{m}$ or $\mathfrak{m}$ belong to $J_{\y, \leq \tilde{d}}(B)$, where $B(\x,\y) = \left(f_1(1,\x,1,\y), \ldots , f_m(1,\x,1,\y)\right)$. This implies $d_{\wit}(B) \leq d_{\tilde{\y}, reg}(B^{(h)})$.
Notice that for the particular case $f(\tilde{\x},\tilde{\y}) = y_{0}^{d_{\tilde{\y},reg}(B)}$, we get that $f(1,\x,1,\y) = 1$ belongs to $J_{\y,\leq \tilde{d}}(B)$. Hence the system $B=\textbf{0}$ has no solution.

Finally, Proposition \ref{prop:deg_reg} implies  
\begin{equation*}
   d_{wit}(B) \leq \left\lceil \frac{n_y(n_x+1)}{m-n_x-1} \right\rceil + 1.
\end{equation*}
\end{proof}

\begin{corollary}\label{coro:condition}
Let $n_x, n_y,m$ be positive integers such that $n_x + n_y \leq m-2$. Then, there is a set $\mathcal{S} \subseteq \B(n_x,n_y,m)$ containing an open Zariski set such that every $B \in \mathcal{S}$ satisfies the following condition: The system $B=\textbf{0}$ has a solution if and only if $1 \notin J_{\y,\leq \tilde{d}}(B)$, with
\begin{equation*}
\tilde{d} = \left\lceil \frac{n_y(n_x+1)}{m-n_x-1} \right\rceil + 1.
\end{equation*}
\end{corollary}

\begin{proof}
   The set $\mathcal{S}$ is the union of the set $O$ from Theorem \ref{teo:witnessbound} and the set of sequences $B$ such that $B=\textbf{0}$ has a solution.
\end{proof}

A computational way know whether $1 \in J_{\y,\tilde{d}}(B)$ for a a given $B \in \B(n_{x},n_{y},m)$, where $n_x + n_y \leq m-2$, is by testing the solvability of the linear system  $\textbf{z} \cdot \textbf{M}_{\y,\leq \tilde{d}}(B) = \textbf{e} $, where $\textbf{e} = \begin{pmatrix} 0 & \cdots & 0 & 1 \end{pmatrix} \in \FF^{n_x \binom{n_y + \tilde{d} -1}{\tilde{d} -1}}$ is a row vector. Such a system has a solution for $\textbf{z}$ if and only if $1 \in J_{\y,\tilde{d}}(B)$. 

\section{Complexity Analysis}\label{sec:classcomplx}
Here we estimate the complexity of solving over-defined generic bilinear systems over finite fields. We propose three slight variants of XL-like algorithms, specifically designed for solving bilinear systems, namely, $\y$-XL, $\y$-MutantXL,  and $\y$-Hybrid. We analyse the complexity of these algorithms and compare them with the efficiency of the F4 algorithm. All over this section $n_{x},n_{y}$ and $m$ are positive integers with  $n_{x} + n_{y} \leq m -2$ and $n_x \leq n_y$. 

\subsection{$\y$-XL}\label{sec:XL}

$\y$-XL is an algorithm for solving a bilinear system $B=\textbf{0}$.
For $B \in  \B(n_x,n_y,m)$, $\y$-XL looks for a linear polynomial in $J_{\y,\leq d}(B)$, for some integer $d$.
%
Provided the existence of a solution, we say $\y$-XL solves the system $B=\textbf{0}$ at degree $d$, if it finds at least one linear equation in $J_{\y,\leq d}(B)$. Otherwise, we say it does not solve the system at degree $d$. The minimum degree $d$ at which $\y$-XL solves a bilinear sequence $B$ is denoted by $\y$-XL$_{sol}(B)$.

\begin{algorithm}
   \caption{$\y$-XL}
    \begin{algorithmic}[1]
      \Function{$\y$-XL}{$B,d$}\Comment{Where $B \in \B(n_x,n_y,m)$ - $d \in \mathbb{Z}^{+}$}
        \State $L = \emptyset $
        \State $\textbf{E} = \text{EchelonForm}\left(\textbf{M}_{\y,\leq d}(B)\right)$
        \State $L = $ Linear polynomials representing rows in $\textbf{E}$.
        \State \textbf{return }L
        \EndFunction
\end{algorithmic}
\end{algorithm}
 
Finding linear polynomials is not the only criterion for deciding whether XL succeeds or not in solving a system, c.f. \citep{Cox:2007:IVA:1204670}. However, studying other termination criteria for $\y$-XL is outside of the scope of this paper.
%
The complexity of $\y$-XL, provided that  $B=\textbf{0}$ has a solution, is upper bounded by the complexity of computing the echelon form of the matrix $\textbf{M}_{\y,\leq d}(B)$, where $d = d_{\y,wit}(B) $ , which is a matrix of size \begin{equation*}
      m\binom{n_{y} + d_{\y,wit}(B)  -2}{d_{\y,wit}(B) -2} \times  (n_x+1)\binom{n_y + d_{\y,wit}(B) -1 }{d_{\y,wit}(B)-1}.
\end{equation*}
%

In most cases, and regardless of whether the system has a solution or not, we can precisely estimate its witness degree.
By Corollary \ref{coro:condition}, for most sequences $B \in \B(n_x,n_y,m)$ with $n_x + n_y \leq m-2$, the complexity of deciding whether or not the system  $B = \textbf{0}$ has a solution (and finding one if it exists) is upper bounded by
\begin{equation*}
 \mathcal{O}\left(m\binom{n_{y} +  \tilde{d} -2}{\tilde{d} -2} \left[n_x \binom{n_y + \tilde{d}-1 }{\tilde{d}-1}\right]^{\omega -1} \right)  \end{equation*}
operations over $\FF$. Here $\tilde{d}$ is as defined in Corollary \ref{coro:condition}, and $\omega$ is the exponent of the complexity of multiplying two square matrices of size $n$.

We experimentally verified this result for $B \in \B(n_x , n_y,m)$ chosen uniformly at random and forced to have a uniform random solution $\textbf{a} \in \FF^{(n_x +1)(n_y +1)}$.
In most cases, linear polynomials can be found by applying $\y$-XL at degree $\tilde{d}$. The results of our experiments are summarized in Tables \ref{tab:tablesummary_1} and \ref{tab:tablesummary_2}.

\subsection{$\y$-MutantXL}\label{sec:yMXL}

This section introduces the mutant variant of $\y$-XL, which we will refer to as $\y$-MXL. The idea here is to apply the same strategy of MutantXL \citep{mxl1,DCabarcas}, but only using monomials involving $\y$ variables. Generally speaking, this strategy consists of taking the degree falls that appear in $\y$-XL, multiplying them by the $\y$-variables and appending them to the set of polynomials. This process is repeated again and again until degree one polynomials appear.

As in the $\y$-XL algorithm, an integer $d$ and  a bilinear sequence $B$ are provided as the input of $\y$-MXL. Similarly, we say that $\y$-MXL solves a system $B=\textbf{0}$ at degree $d$ (provided a solution exists), if it finds linear polynomials. We denote by $\y$-MXL$_{sol}(B)$ the minimum integer $d$ at which $\y$-MXL solves the system $B = \textbf{0}$.

The advantage of $\y$-MXL over $\y$-XL is that the mutant version might finish at degree $d_{\y-ff}(B)$ (or not far from it), which is in general smaller than $d_{\y, wit}(B)$. Moreover, following Section \ref{sec:firstfall}, we can precisely estimate $d_{\y-ff}(B)$ for most bilinear sequences  as
\begin{equation*}  
    T_{ff}(n_x,n_y,m) = \min \left\lbrace d \in \mathbb{Z} ^{+} \; \vert  \;  d > \frac{n_x(n_y-1)}{m-n_x} +1\right\rbrace .
\end{equation*}

Precisely estimating $\y$-MXL$_{sol}(B)$ is an open question.
From the experimental data showed in Tables \ref{tab:tablesummary_1} and \ref{tab:tablesummary_2}, we conjecture that if  $B \in  \B(n_x,n_y,m)$ is $\y$-semiregular and  $ T_{ff}(n_x,n_y,m)  =  T_{ff}(n_x,n_y,m-1)$, then $\y$-MXL$_{sol}(B)=d_{\y-ff}(B)$.  
This conjecture is reasonable because when  $T_{ff}$ is equal for $(n_x,n_y,m)$ and $(n_x,n_y,m-1)$, the number of degree falls is substantial. Therefore, 
 the complexity of solving most systems $B=\textbf{0}$ using $\y$-MXL, where
$B \in \B(n_{x},n_{y},m)$ and $ T_{ff}(n_x,n_y,m)  =  T_{ff}(n_x,n_y,m-1)$ is given by
 \begin{equation*}
    \mathcal{O}\left(m \binom{n_{y} + d_{\y -ff}(B)-2}{d_{\y -ff}(B)-2} \left[n_{x}\binom{n_{y} + d_{\y -ff}(B)-1}{d_{\y-ff}(B)-1}\right]^{\omega-1} \right)
\end{equation*}
multiplications over $\FF$, where $\omega$ is the linear algebra constant and
$ d_{\y - ff}(B) = T_{ff}(n_x,n_y,m) $.

\subsection{$\y$-Hybrid Approach}\label{sec:HXL}

In this section we describe and analyze the complexity of a hybrid algorithm for solving the system $B=\textbf{0}$ for a bilinear sequence $B$. Hybrid approaches for solving generic systems of bilinear equations have been studied by different researchers \citep{booleansolve,hybrid_BFP09,hybrid_improved_BettaleFP12}.
The general idea is to try all possible values for some variables and check the consistency of the resulting partial evaluation.

Throughout this section $\x = (x_{1},\ldots , x_{n_x})$ and $\y = (y_{1},\ldots , y_{n_y})$ are enumerated sets of variables. For integers $a_x$, $a_y$, with  $0\leq a_{x} < n_x$ and $0\leq a_{y} < n_y$, we use $\x_{a_x}$ and $\y_{a_y}$ to denote the vectors of variables $(x_{a_x + 1}, \ldots , x_{n_x})$ and $(y_{a_y + 1}, \ldots , y_{n_y})$, respectively. 

\begin{definition}[Partial Evaluation]\label{def:partial}
Let $B(\x,\y) \in \B(n_x,n_y,m)$, $\textbf{u} =(u_1,\ldots,u_{a_x}) \in \FF^{a_x}$ and $\textbf{v} =(v_1,\ldots,v_{a_y})\in \FF^{a_y}$, where $0\leq a_{x} < n_x$ and $0\leq a_{y} < n_y$. We use $B_{(\textbf{u},\textbf{v})}(\x_{a_x},\y_{a_y})$ to denote the bilinear sequence in $\B(n_x - a_x,n_y -a_y,m)$ given by 
$$ B(u_1,\ldots, u_{a_x},x_{a_x +1},\ldots x_{n_x},v_1,\ldots, v_{a_y},y_{a_y +1},\ldots y_{n_y} ).$$
The sequence $B_{(\textbf{u},\textbf{v})}(\x_{a_x},\y_{a_y})$ is called the partial evaluation of $B$ in $(\textbf{u},\textbf{v})$. For short we use $B_{(\textbf{u},\textbf{v})}$, when the involved variables $\x$, $\y$, $\x_{a_x}$ and $\y_{a_y}$ are clear in the context. 
\end{definition}

Given $B \in \B(n_x,n_y,m)$, the $\y$-HXL$_{a_x,a_y}$ algorithm goes through all pairs of vectors $(\textbf{u},\textbf{v}) \in \FF^{a_x} \times \FF^{a_y} $, and checks the consistency of the partially evaluated bilinear system $B_{(\textbf{u},\textbf{v})}{(\x_{a_x},\y_{a_y})} = \textbf{0}$. It stops when it finds a system $B_{(\textbf{u},\textbf{v})} = \textbf{0}$ being consistent.
This procedure is described in Algorithm \ref{alg:hxl}. This computes a partial solution of the system $B=\textbf{0}$.
It can then be applied recursively until a whole vector $\textbf{u} \in \FF^{n_x}$ or $\textbf{v} \in \FF^{n_y}$ is found, such that the system $B_{(\textbf{u},\textbf{v})} = \textbf{0}$ has a solution.
After this, a complete solution can be found by solving a linear system in the remaining unknown variables.

\begin{algorithm}
   \caption{$\y$\textbf{-HXL}$_{a_x,a_y}$}\label{alg:hxl}
    \begin{algorithmic}[1]
      \Function{$\y$\textbf{-HXL}$_{a_x,a_y}$}{$B$}\Comment{Where $B \in \B(n_x,n_y,m)$}
              \State{ $ d = \left\lceil \frac{(n_y-a_y)(n_x - a_x+1)}{m-n_x +a_x-1} \right\rceil +1$}\label{alg2:step2}
              \State{$\textbf{e} = \begin{pmatrix} 0 & \cdots & 0 & 1 \end{pmatrix} \in \FF^{n_x \binom{n_y + d -1}{d -1}}$}
      \For{$(\textbf{u}, \textbf{v}) \in \FF ^{a_x}\times \FF ^{a_y} $}
        \If{ $\textbf{z} \cdot \textbf{M}_{\y,\leq d}\left(B_
        {(\textbf{u},\textbf{v})}\right) = \textbf{e}$ is inconsistent} \label{alg2:step4}
        \State{\textbf{return}  $(\textbf{u},\textbf{v})$} 
        \EndIf
    \EndFor
    \EndFunction
\end{algorithmic}
\end{algorithm}


For $(\textbf{u},\textbf{v}) \in \FF^{a_x} \times  \FF^{a_y}$ define a random variable $\mathcal{X}_{(\textbf{u},\textbf{v})}$ taking values in the set of bilinear equations $\B(n_x-a_x,n_y-a_y,m)$.
In each realization of $\mathcal{X}_{(\textbf{u},\textbf{v})}$ a sequence $B(\x,\y) \in \B(n_x,n_y,m)$ is chosen and the output of $\mathcal{X}_{(\textbf{u},\textbf{v})}$ is $B_{(\textbf{u},\textbf{v})}$.
If $B \in \B(n_x,n_y,m)$ is chosen uniformly at random,  then the random variable $\mathcal{X}_{(\textbf{u},\textbf{v})}$ is uniform in $\B(n_x-a_x,n_y-a_y,m)$.


By Corollary \ref{coro:condition}, when $n_x + n_y \leq m-2$, there is a subset $\mathcal{S} \subset \B(n_x - a_x,n_y - a_y,m)$ containing an open Zariski set, such that for every $B_{(\textbf{u},\textbf{v})} \in \mathcal{S}$, the consistency (or inconsistency) of the system $B_{(\textbf{u},\textbf{v})} =\textbf{0}$ can be verified by checking the inconsistency (or consistency) of the linear system $\textbf{z} \cdot \textbf{M}_{\y,\leq \tilde{d}}(B_{(\textbf{u},\textbf{v})}) = \textbf{e}$, where
\begin{equation}\label{eq:dtilde}
\tilde{d} = \left\lceil \frac{(n_y-a_y)(n_x - a_x+1)}{m-n_x +a_x-1} \right\rceil +1.
\end{equation}
Moreover, by the same result and since $\mathcal{X}_{(\textbf{u},\textbf{v})}$ is a uniform random variable, the $\y$-witness degree of the partially evaluated sequence is expected to be upper bounded by $\tilde{d}$. That is why this is the value chosen for $d$ in step \ref{alg2:step2} of Algorithm \ref{alg:hxl}.

It is possible and advantageous to use Wiedemann's Algorithm to check the consistency of the linear system in  step \ref{alg2:step4} of Algorithm \ref{alg:hxl}.
Given a matrix $\textbf{A}$ with coordinates in $\FF$, and provided that the system $\textbf{z}\cdot \textbf{A} = \textbf{b}$ has at least a solution for $\textbf{z}$, Wiedemann's Algorithm returns a solution by performing an expected number of operations over $\FF$ upper bounded by $\mathcal{O}\left(n_0 (\eta + n_1\log n_1) \log n_0 \right)$. Here $\eta$ is the number of non-zero entries in $\textbf{A}$, and $n_0,n_1$ are the minimum and maximum between the number of rows and columns of $\textbf{A}$, respectively \citep{wiedemann}.

For $a_x, a_y$ fixed, the complexity of $\y$-HLX$_{a_x,a_y}$ (using Wiedemann's algorithm) is given by 
\begin{align*}
    \mathcal{O}\left(q^{a_x + a_y} (n_y - a_y+1)(n_x - a_x +1)^{3}\binom{n_y - a_y + \tilde{d} -1}{\tilde{d} -1}^{2} \right).
\end{align*}

For $a_x, a_y$ fixed, the complexity of  $\y$-HXL$_{a_x,a_y}$ (using Gaussian elimination) is given by 
\begin{align*}
 \mathcal{O}\left(q^{a_x + a_y} m\binom{n_{y} -a_y +  \tilde{d} -2}{\tilde{d} -2} \left[(n_x -a_x +1) \binom{n_y -a_y  + \tilde{d}-1 }{\tilde{d}-1}\right]^{\omega -1} \right), 
\end{align*}
where $\tilde{d}$ is the integer defined in Equation \eqref{eq:dtilde}.

Computing asymptotic formulas for the values $a_{x}$ and $a_{y}$, which lead to an optimal complexity of $\y$-HXL, is out of the scope of this paper. Instead, we use the complexity formulas to find, numerically, the optimal values of $a_x$ and $a_y$ for a given set of parameters $q,m,n_x,n_y$. The results are shown in Table \ref{tab:complexity_comparison}. We also compare this optimal complexity of $\y$-HXL with the complexity of $\y$-MXL, see Section \ref{sec:yMXL}. We can see that in most of the cases $\y$-HXL with the optimal $a_x,a_y$ outperforms $\y$-MXL.  

\subsection{Out-of-the-Box Methods}\label{sec:F4}

In order to solve a bilinear systems, it is also possible to use an out-of-the-box algorithm, for example applying the F4 algorithm or trying all possible values of the $\x$-variables.
Here, we estimate the complexity of this approach, as a point of reference for comparison.

As mentioned in Section \ref{sec:poly_solving}, the complexity of solving a system of polynomial equations $B=\textbf{0}$, where $B$ is semiregular, can be estimated by using the first degree fall. 
Bilinear sequences in $\FF[\x,\y]$ are not s, they form a relative small set of polynomials in the set of all  quadratic polynomials in $\FF[\x,\y]$. Thus we would not expect that the first fall degree $d_{ff}(B)$ of a bilinear sequence $B$ is also its solving degree. However, in all the experiments we were able to conduct, this was the case  (see Tables \ref{tab:tablesummary_1} and \ref{tab:tablesummary_2}). Therefore, we estimate that the complexity of solving a bilinear system $B = \textbf{0}$ using F4, where $B \in \B(n_x,n_y,m)$ is $\y$-semiregular, is given by
\begin{equation*}
    \mathcal{O}\left(\binom{n_{x} + n_{y} + d_{ff}(B)}{d_{ff}(B)}^{\omega}\right).
\end{equation*}

Another way to solve a bilinear system is by trying all possible values of the variables from one set ($\x$ or $\y$) and then check the consistency of the remaining linear system. Since $n_x \leq n_y$, it is better to test all $\x$ variables. Notice this simple method can be seen as a special case of the algorithm $\y$-HXL$_{a_x,a_y}$ when $a_x = n_x $ and $a_y = 0$. The complexity of this method is
\begin{equation*}
      \mathcal{O}\left(q^{n_x} m n_y^{\omega -1} \right).
\end{equation*}
\section{Experimental Results}\label{sec:experimental}

In this section we show some experimental results that confirm the theoretical findings of the paper, illustrate some of the results, and fill in some gaps.

In order to evaluate the validity of Conjecture \ref{conj:generic_semiregular}, we performed the experiment described in Algorithm \ref{alg:semiregulairy_test}, whose results are presented in Table \ref{tab:conjecture}.
They show that with high probability, a randomly chosen bilinear sequence is $\y$-semiregular, supporting the validity of the conjecture.

\begin{algorithm}[H]
\caption{Randomly Testing $\y$-Semiregulary}\label{alg:semiregulairy_test}
\textbf{Input: }Positive integers $n_x,n_y,m$ such that $n_x + n_y \leq m$\\
\textbf{Output: }$\mathsf{True}$, if a randomly chosen bilinear sequence is $\y$-semiregular. $\mathsf{False}$, otherwise. 
\begin{algorithmic}[1]
\State{$B \leftarrow \B^{(h)}(n_x,n_y,m)$}\vspace{0.5em} \Comment{Uniformly at random}
\State{$\tilde{d} \leftarrow \left\lceil \frac{n_x(n_y-1)}{m-n_x}\right\rceil +1$} \vspace{0.5em}
\State{$r \leftarrow \binom{n_y + \tilde{d}-3}{\tilde{d}-3} + n_x \binom{n_y  + \tilde{d} -2}{\tilde{d}-1}$}
\vspace{0.5em}
\State{$\textbf{M} \leftarrow \textbf{M}_{\y,\leq \tilde{d}}(B)$}\vspace{0.5em}
\If{\rank(\textbf{M}) = r} \label{testing_step5}
    \State{\textbf{return $\mathsf{True}$}}
    \Else 
        \State{\textbf{return} $\mathsf{False}$}
\EndIf
\end{algorithmic}
\end{algorithm}

Algorithm \ref{alg:semiregulairy_test} indeed checks whether a sequence $B\in \B^{(h)}(n_x,n_y,m)$ is $\y$-semiregular.
In general, for every $2\leq j < \tilde{d}$, we have that $\rank \left( \textbf{M}_{\y,j}(B)\right) \leq  m \binom{n_y + j-3}{j-2} $ and   $\rank \left( \textbf{M}_{\y,\tilde{d}}(B)\right) \leq  n_x \binom{n_y + \tilde{d}-2}{\tilde{d}-1}$ (see Section \ref{sec:jac_syzygies}).  
In particular, when $B$ is homogeneous, the rank of the whole $\y$-Macaulay matrix is given by
$ \rank\left(\textbf{M}_{\y,\leq \tilde{d}}(B)\right) = \sum _{j=1}^{\tilde{d}}\rank \left( \textbf{M}_{\y,j}(B)\right).$
Thus, when the condition in step \ref{testing_step5} is satisfied, we guarantee that for each $j =1,\ldots ,\tilde{d}-1$
$$ \rank \left( \textbf{M}_{\y,j}(B)\right) = m \binom{n_y  + j -3}{j-2},$$
and at the same time 
$$ \rank \left( \textbf{M}_{\y,\tilde{d}}(B)\right) =n_x \binom{n_y  + \tilde{d} -2}{\tilde{d}-1}, $$
which means $B$ is $\y$-semiregular.

Table \ref{tab:conjecture} shows notable variations in the percentage of $\y$-semiregular sequences across different parameters. This phenomenon deserves some explanation.
For every choice of parameters such that $(m-n_x)$ does not divide $n_x(n_y -1)$, the probability that a randomly chosen sequence $B \in \B(n_x,n_y,m)$ happens to be $\y$-semiregular is overwhelming. In the other cases, when $(m-n_x)$ does divide $n_x(n_y -1)$, that probability is at least $85\%$. This difference is because, when $n_x(n_y -1)/(m-n_x)$ is an integer, we have that for $\tilde{d} =  n_x(n_y -1)/(m-n_x) +1$
$$ n_x \binom{n_y + \tilde{d}-2 }{\tilde{d}-1} = m \binom{n_y + \tilde{d}-3 }{\tilde{d}-2}, $$
and thus, if $B \in \B^{(h)}(n_x,n_y,m)$ the submatrix $\textbf{M}_{\y,\tilde{d}}(B)$ is a square matrix.
Since the coefficients of this matrix are on a finite field, we expect it to be invertible with near-one but non-overwhelming probability in the size of the field $\FF$.

\begin{table}[H]
\centering
\caption{Experimental results to verify Conjecture \ref{conj:generic_semiregular}. Here we use $n_x =4$, the column $\%$ shows the percentage of times the chosen sequence $B \in \B^{(h)}(n_x,n_y,m)$ was $\y$-semiregular, and $\tilde{d}$, which is equal to $\lceil n_x(n_y-1)/(m-n_x)\rceil +1$, indicates the $\y$-degree of regularity according to Proposition \ref{prop:deg_reg}. The random experiment described in Algorithm \ref{alg:semiregulairy_test} was executed a hundred of times for every choice of parameters.}\label{tab:conjecture}. 
\begin{tabular}[h]{llllllllllll}
\toprule
$n_y$ & $m$ & $\tilde{d}$ & $\%$ & $n_y$ & $m$ & $\tilde{d}$ & $\%$ & $n_y$ & $m$ & $\tilde{d}$ & $\%$ \\ \midrule
\multirow{9}{*}{4} & 8 & 4 & 88 & 5 & 18 & 3 & 100 & \multirow{6}{*}{7} & 17 & 3 & 100 \\ \cmidrule{2-4}\cmidrule{5-8}\cmidrule{10-12}
 & 9 & 4 & 100 & \multirow{11}{*}{6} & 10 & 5 & 99 &  & 18 & 3 & 100 \\  \cmidrule{2-4}\cmidrule{6-8}\cmidrule{10-12}
 & 10 & 3 & 93 &  & 11 & 4 & 100 &  & 19 & 3 & 100 \\ \cmidrule{2-4}\cmidrule{6-8}\cmidrule{10-12}
 & 11 & 3 & 100 &  & 12 & 4 & 100 &  & 20 & 3 & 100 \\  \cmidrule{2-4}\cmidrule{6-8}\cmidrule{10-12}
 & 12 & 3 & 100 &  & 13 & 4 & 100 &  & 21 & 3 & 100 \\  \cmidrule{2-4}\cmidrule{6-8}\cmidrule{10-12}
 & 13 & 3 & 100 &  & 14 & 3 & 92 &  & 22 & 3 & 100 \\  \cmidrule{2-4}\cmidrule{6-8}\cmidrule{9-12}
 & 14 & 3 & 100 &  & 15 & 3 & 100 & \multirow{12}{*}{8} & 12 & 5 & 98 \\   \cmidrule{2-4}\cmidrule{6-8}\cmidrule{10-12}
 & 15 & 3 & 99 &  & 16 & 3 & 100 &  & 13 & 5 & 100 \\   \cmidrule{2-4}\cmidrule{6-8}\cmidrule{10-12}
 & 16 & 2 & 93 &  & 17 & 3 & 100 &  & 14 & 4 & 100 \\   \cmidrule{1-4}\cmidrule{6-8}\cmidrule{10-12}
 & 9 & 5 & 99 &  & 18 & 3 & 100 &  & 15 & 4 & 100 \\   \cmidrule{2-4}\cmidrule{6-8}\cmidrule{10-12}
\multirow{9}{*}{5} & 10 & 4 & 100 &  & 19 & 3 & 100 &  & 16 & 4 & 100 \\   \cmidrule{2-4}\cmidrule{6-8}\cmidrule{10-12}
 & 11 & 4 & 100 &  & 20 & 3 & 100 &  & 17 & 4 & 100 \\   \cmidrule{2-4}\cmidrule{5-8}\cmidrule{10-12}
 & 12 & 3 & 90 & \multirow{6}{*}{7} & 11 & 5 & 100 &  & 18 & 3 & 92 \\   \cmidrule{2-4}\cmidrule{6-8}\cmidrule{10-12}
 & 13 & 3 & 100 &  & 12 & 4 & 86 &  & 19 & 3 & 100 \\   \cmidrule{2-4}\cmidrule{6-8}\cmidrule{10-12}
 & 14 & 3 & 100 &  & 13 & 4 & 100 &  & 20 & 3 & 100 \\   \cmidrule{2-4}\cmidrule{6-8}\cmidrule{10-12}
 & 15 & 3 & 100 &  & 14 & 4 & 100 &  & 21 & 3 & 100 \\   \cmidrule{2-4}\cmidrule{6-8}\cmidrule{10-12}
 & 16 & 3 & 100 &  & 15 & 4 & 100 &  & 22 & 3 & 100 \\   \cmidrule{2-4}\cmidrule{6-8}\cmidrule{10-12}
 & 17 & 3 & 100 &  & 16 & 3 & 88 &  & 23 & 3 & 100 \\ \bottomrule
\end{tabular}
\end{table}

Tables \ref{tab:tablesummary_1} and \ref{tab:tablesummary_2} serve to compare our theoretical estimates with experimental results for the various algorithms on randomly chosen bilinear sequences.
It is worth comparing the solving degree of $\y$-XL and $\y$-MXL, which is the minimum degree where the algorithms find linear polynomials in $I_{\y}$ (see sections \ref{sec:XL} and \ref{sec:yMXL}).
It is known that, for every $B\in\B(n_x,n_y,m)$, with $n_x + n_y \leq m-2$, we have $\y\text{-MXL}_{\text{sol}}(B) \leq \y\text{-XL}_{\text{sol}}(B)$.  However, the inequality might be strict in some cases. The tables show that this is indeed the case for some parameters.
Notice that for certain parameters $\y\text{-MXL}_{\text{sol}}(B) = \y\text{-XL}_{\text{sol}}(B)-1$.
For such parameters $\y$-MXL shows an exponential speed up over $\y-$XL.

It is also known that $d_{ff}$ is upper bounded by $ d_{\y-ff}$.
Yet, it might be the case that for some parameters, $d_{ff}$ is strictly less than $d_{\y-ff}$.
However, for every single instance we ran, with $n_x + n_y \leq m$, we observed $d_{ff} = d_{\y-ff}$ and this was also the solving degree of F4.
Based on this, we conjecture that $d_{\y-ff}$ is a tight upper bound for the solving degree of F4 for $\y$-semiregular sequences. Thus, the complexity of solving a system $B = \textbf{0}$, where $B \in \B(n_x,n_y,m)$, is given by
\begin{equation*}
    \mathcal{O}\left(\binom{n_{x} + n_{y} + d_{\y-ff}(B)}{d_{\y-ff}(B)}^{\omega}\right).
\end{equation*}
This is less efficient than $\y$-MXL, provided that $ T_{ff}(n_x,n_y,m)  =  T_{ff}(n_x,n_y,m-1)$, see Section \ref{sec:yMXL}.

\begin{table}[H]
\centering
\caption{Comparison between the solving degrees for $\y$-XL, $\y$-MXL and F4, the $\y$-first fall degree and the first fall degree of a sequence $B \in \B(4,n_y,m)$ chosen uniformly at random over $GF(13)$. The column $T_{wit}$ shows the theoretical upper bound for the solving degree of $\y$-XL given in Theorem \ref{teo:witnessbound} and $T_{ff}$ is the theoretical upper bound for $d_{\y -ff}$ given by Conjecture \ref{conj:generic_semiregular}, as explained in Section \ref{sec:yMXL}. For each sequence $B$, $d_{\y -ff}$ and F4$_{ff}$ are, respectively, the $\y$-first fall degree and the first fall degree of $B$.
$\y$-XL$_{sol}$ and $\y$-MXL$_{sol}$ are the solving degree in $\y$-XL and in $\y$-MXL of $B$, respectively, as defined in Sections \ref{sec:XL} and \ref{sec:yMXL}. F4$_{sol}$ is the maximum degree reached during the Gr\"{o}bner basis computation of the ideal generated by $B$. The five rightmost columns show the most common value  obtained for each set of parameters out of a hundred random instances, and the value in parenthesis indicates the corresponding relative frequency. }\label{tab:tablesummary_1}

\begin{tabular}{ccccccccc}
\toprule
$n_y$ & $m$ & $T_{ff}$ & $T_{wit}$ & $d_{\y-ff}$  & $\y$-XL$_{sol}$ & $\y$-MXL$_{sol}$ & F4$_{ff}$ & F4$_{sol}$ \\ \midrule
& 10 & 4 & 5 & 4 (0.93) & 5 (0.89) & 4 (1.0) & 4 (0.87) & 4 (0.99) \\ \cmidrule{2-9}
 \multirow{7}{*}{4}  & 11 & 3 & 5 & 3 (1.0) & 5 (1.0) & 4 (1.0) & 3 (1.0) & 3 (1.0) \\ \cmidrule{2-9}
 & 12 & 3 & 4 & 3 (1.0) & 4 (1.0) & 3 (1.0) & 3 (1.0) & 3 (1.0) \\ \cmidrule{2-9}
 & 13 & 3 & 4 & 3 (1.0) & 4 (1.0) & 3 (1.0) & 3 (1.0) & 3 (1.0) \\ \cmidrule{2-9}
 & 14 & 3 & 4 & 3 (1.0) & 4 (1.0) & 3 (1.0) & 3 (1.0) & 3 (1.0) \\ \cmidrule{2-9}
 & 15 & 3 & 3 & 3 (1.0) & 3 (0.89) & 3 (1.0) & 3 (1.0) & 3 (1.0) \\ \cmidrule{2-9}
 & 16 & 3 & 3 & 3 (1.0) & 3 (1.0) & 3 (1.0) & 3 (0.89) & 3 (1.0) \\ \midrule
 & 11 & 4 & 6 & 4 (1.0) & 6 (0.98) & 4 (0.98) & 4 (1.0) & 4 (1.0) \\ \cmidrule{2-9}
 \multirow{8}{*}{5} & 12 & 4 & 5 & 4 (0.95) & 5 (1.0) & 4 (1.0) & 4 (0.96) & 4 (1.0) \\ \cmidrule{2-9}
 & 13 & 3 & 5 & 3 (1.0) & 5 (1.0) & 4 (1.0) & 3 (1.0) & 3 (1.0) \\ \cmidrule{2-9}
 & 14 & 3 & 4 & 3 (1.0) & 4 (1.0) & 3 (1.0) & 3 (1.0) & 3 (1.0) \\ \cmidrule{2-9}
 & 15 & 3 & 4 & 3 (1.0) & 4 (1.0) & 3 (1.0) & 3 (1.0) & 3 (1.0) \\ \cmidrule{2-9}
 & 16 & 3 & 4 & 3 (1.0) & 4 (1.0) & 3 (1.0) & 3 (1.0) & 3 (1.0) \\ \cmidrule{2-9}
 & 17 & 3 & 4 & 3 (1.0) & 4 (1.0) & 3 (1.0) & 3 (1.0) & 3 (1.0) \\ \cmidrule{2-9}
 & 18 & 3 & 3 & 3 (1.0) & 3 (1.0) & 3 (1.0) & 3 (1.0) & 3 (1.0) \\ \midrule
 & 12 & 4 & 6 & 4 (1.0) & 6 (0.99) & 4 (0.99) & 4 (1.0) & 4 (1.0) \\ \cmidrule{2-9}
 \multirow{9}{*}{6} & 13 & 4 & 5 & 4 (1.0) & 5 (1.0) & 4 (1.0) & 4 (1.0) & 4 (1.0) \\ \cmidrule{2-9}
 & 14 & 4 & 5 & 4 (0.94) & 5 (1.0) & 4 (1.0) & 4 (0.95) & 4 (1.0) \\ \cmidrule{2-9}
 & 15 & 3 & 4 & 3 (1.0) & 4 (0.95) & 4 (1.0) & 3 (1.0) & 3 (1.0) \\ \cmidrule{2-9}
 & 16 & 3 & 4 & 3 (1.0) & 4 (1.0) & 3 (1.0) & 3 (1.0) & 3 (1.0) \\ \cmidrule{2-9}
 & 17 & 3 & 4 & 3 (1.0) & 4 (1.0) & 3 (1.0) & 3 (1.0) & 3 (1.0) \\ \cmidrule{2-9}
 & 18 & 3 & 4 & 3 (1.0) & 4 (1.0) & 3 (1.0) & 3 (1.0) & 3 (1.0) \\ \cmidrule{2-9}
 & 19 & 3 & 4 & 3 (1.0) & 4 (1.0) & 3 (1.0) & 3 (1.0) & 3 (1.0) \\ \cmidrule{2-9}
 & 20 & 3 & 3 & 3 (1.0) & 3 (0.95) & 3 (1.0) & 3 (1.0) & 3 (1.0) \\ \bottomrule

\end{tabular}
\end{table}

\begin{table}[H]
\centering
\caption{Comparison between the solving degrees for $\y$-XL, $\y$-MXL and F4, the $\y$-first fall degree and the first fall degree of a sequence $B \in \B(4,n_y,m)$ chosen uniformly at random over $GF(13)$. The column $T_{wit}$ shows the theoretical upper bound for solving degree of $\y$-XL  given in Theorem \ref{teo:witnessbound} and  $T_{ff}$ is the theoretical upper bound for $d_{\y -ff}$ given by Conjecture \ref{conj:generic_semiregular}, as explained in Section \ref{sec:yMXL}. For each sequence $B$, $d_{\y -ff}$ and F4$_{ff}$ are, respectively, the $\y$-first fall degree and the first fall degree of $B$.
$\y$-XL$_{sol}$ and $\y$-MXL$_{sol}$ are the solving degree in $\y$-XL and in $\y$-MXL of $B$, respectively, as defined in the sections \ref{sec:XL} and \ref{sec:yMXL}. F4$_{sol}$ is the maximum degree reached during the Gr\"{o}bner basis computation of the ideal generated by $B$. The five rightmost columns shows the most common value  obtained for each set of parameters out of hundred of realizations, the value in parenthesis indicates the corresponding the relative frequency.}\label{tab:tablesummary_2}

\begin{tabular}{ccccccccc}
\toprule
$n_y$ & $m$ & $T_{ff}$ & $T_{wit}$ & $d_{\y-ff}(B)$  & $\y$-XL$_{sol}$ & $\y$-MXL$_{sol}$ & F4$_{ff}$ & F4$_{sol}$ \\ \midrule
  & 13 & 4 & 6 & 4 (1.0) & 6 (1.0) & 5 (1.0) & 4 (1.0) & 4 (1.0) \\ \cmidrule{2-9}
 \multirow{10}{*}{7}  & 14 & 4 & 5 & 4 (1.0) & 5 (1.0) & 4 (1.0) & 4 (1.0) & 4 (1.0) \\ \cmidrule{2-9}
 & 15 & 4 & 5 & 4 (1.0) & 5 (1.0) & 4 (1.0) & 4 (1.0) & 4 (1.0) \\ \cmidrule{2-9}
 & 16 & 4 & 5 & 4 (0.94) & 5 (1.0) & 4 (1.0) & 4 (0.95) & 4 (1.0) \\ \cmidrule{2-9}
 & 17 & 3 & 4 & 3 (1.0) & 4 (1.0) & 4 (1.0) & 3 (1.0) & 3 (1.0) \\ \cmidrule{2-9}
 & 18 & 3 & 4 & 3 (1.0) & 4 (1.0) & 3 (1.0) & 3 (1.0) & 3 (1.0) \\ \cmidrule{2-9}
 & 19 & 3 & 4 & 3 (1.0) & 4 (1.0) & 3 (1.0) & 3 (1.0) & 3 (1.0) \\ \cmidrule{2-9}
 & 20 & 3 & 4 & 3 (1.0) & 4 (1.0) & 3 (1.0) & 3 (1.0) & 3 (1.0) \\ \cmidrule{2-9}
 & 21 & 3 & 4 & 3 (1.0) & 4 (1.0) & 3 (1.0) & 3 (1.0) & 3 (1.0) \\ \cmidrule{2-9}
 & 22 & 3 & 4 & 3 (1.0) & 4 (1.0) & 3 (1.0) & 3 (1.0) & 3 (1.0) \\ \midrule
 & 14 & 4 & 6 & 4 (1.0) & 6 (1.0) & 5 (1.0) & 4 (1.0)& 4 (1.0) \\ \cmidrule{2-9}
 & 15 & 4 & 5 & 4 (1.0) & 5 (0.9) & 4 (1.0) & 4 (1.0) & 4 (1.0) \\ \cmidrule{2-9}
 \multirow{11}{*}{8} & 16 & 4 & 5 & 4 (1.0) & 5 (1.0) & 4 (1.0) & 4 (1.0) & 4 (1.0) \\ \cmidrule{2-9}
 & 17 & 4 & 5 & 4 (1.0) & 5 (1.0) & 4 (1.0) & 4 (1.0) & 4 (1.0) \\ \cmidrule{2-9}
 & 18 & 4 & 5 & 4 (0.91) & 5 (1.0) & 4 (1.0) & 4 (0.92) & 4 (1.0) \\ \cmidrule{2-9}
 & 19 & 3 & 4 & 3 (1.0) & 4 (1.0) & 4 (1.0) & 3 (1.0) & 3 (1.0) \\ \cmidrule{2-9}
 & 20 & 3 & 4 & 3 (1.0) & 4 (1.0) & 3 (1.0) & 3 (1.0) & 3 (1.0) \\ \cmidrule{2-9}
 & 21 & 3 & 4 & 3 (1.0) & 4 (1.0) & 3 (1.0) & 3 (1.0) & 3 (1.0) \\ \cmidrule{2-9}
 & 22 & 3 & 4 & 3 (1.0) & 4 (1.0) & 3 (1.0) & 3 (1.0) & 3 (1.0) \\ \cmidrule{2-9}
 & 23 & 3 & 4 & 3 (1.0) & 4 (1.0) & 3 (1.0) & 3 (1.0) & 3 (1.0) \\ \cmidrule{2-9}
 & 24 & 3 & 4 & 3 (1.0) & 4 (1.0) & 3 (1.0) & 3 (1.0) & 3 (1.0) \\ \bottomrule
\end{tabular}
\end{table}


Based on the complexity estimates in Section \ref{sec:classcomplx}, we compare the complexity of $\y$-MXL and $\y$-HXL for different parameters.
Table \ref{tab:complexity_comparison} illustrates some of the trends.
In the case of $\y$-HXL, for each set of parameters, we report the optimal number of variables to guess and the optimal linear algebra algorithm between Strassen's and Wiedemann's.
In the case of $\y$-MXL, we accept the conjecture that $\y$-MXL$_{sol}(B) = d_{\y,ff}(B)$ for $B \in  \B(n_x,n_y,m)$ provided $ T_{ff}(n_x,n_y,m)  =  T_{ff}(n_x,n_y,m-1)$.

By far, $\y$-HXL outperforms $\y$-MXL. The advantage of $\y$-HXL is specially acute for smaller values of $n_y$, $q$, and $m$, but still significant for larger values.

It is worth noting the behavior of the optimal number of variables to guess in $\y$-MXL.
For small fields, it is better to guess most $\x$-variables. As the size of the field grows, guessing obviously becomes more expensive. Also, as $m$ grows, guessing becomes less attractive, because the witness degree becomes smaller, thus checking consistency becomes less expensive. However, this tendency is less pronounced for larger values of $n_y$, because the witness degree is proportional to $n_y$.

\begin{table}[H]
    \centering
    \caption{Complexity estimates comparison between  $\y$-MXL and $\y$-HXL$_{a_x,a_y}$ for $n_x = 20$, $n_y = 20,30 $ and different values of $q$ and $m$. The columns $MXL$ and $HXL$ indicate the complexity of $\y$-MXL algorithm and $\y$-HXL$_{a_x,a_y}$, respectively. They are computed as $\log _{2}( \ast)$, where $\ast$ are the complexity estimates in Section \ref{sec:classcomplx} for given values $q,n_y,a_x,a_y$.The values $a_x$ and $a_y$ are the ones providing an optimal complexity in $\y$-HXL$_{a_x,a_y}$. The column $Alg$ indicates the linear algebra algorithm given better complexity in $\y$-HXL$_{a_x,a_y}$, 'S' means Strassen's Algorithm while 'W' means Widemann's Algorithm, we use $\omega = 2.8$ in this case. 
    }\label{tab:complexity_comparison}
    \begin{tabular}{ccccccc@{\hskip 0.4in}cccccc}
    \toprule
    $n_y$ &  \multicolumn{6}{c}{ 20}   &  \multicolumn{6}{c}{30}  \\ \hline
        $q$  & $m$ & $MXL$ & $a_x$ & $a_y$ & $HXL$ & $Alg$ & $m$ & $MXL$ & $a_x$ & $a_y$ & $HXL$ & $Alg$ \\ \midrule
         &   42&110&19&0&59&S &52&136&20&0&61&S \\ \cmidrule{2-13}
         \multirow{6}{*}{5} &  46&101&19&0&59&S &56&128&20&0&61&S \\ \cmidrule{2-13}
         &  50&94&19&0&60&S &60&119&20&0&61&S \\ \cmidrule{2-13}
         &  54&90&20&0&60&W &64&115&19&0&61&S  \\ \cmidrule{2-13}
         &  58&86&20&0&60&W &68&110&19&0&61&S \\ \cmidrule{2-13}
         &  62&82&20&0&60&W &72&106&19&0&61&S  \\ \midrule
         &  42&110&19&0&85&S &52&136&20&0&89&S \\ \cmidrule{2-13}
        \multirow{6}{*}{13} &  46&101&19&0&86&S &56&128&20&0&89&S \\ \cmidrule{2-13}
         &  50&94&2&1&86&W &60&119&20&0&89&S \\ \cmidrule{2-13}
         & 54&90&3&0&80&W &64&115&19&0&87&S \\ \cmidrule{2-13}
         &  58&86&3&0&77&W &68&110&19&0&87&S  \\ \cmidrule{2-13}
         &  62&82&2&0&74&W &72&106&19&0&87&S \\ \midrule
         & 42&110&3&0&98&W &52&136&20&0&114&S  \\ \cmidrule{2-13}
         \multirow{6}{*}{31} & 46&101&1&0&92&W &56&128&1&0&110&W \\ \cmidrule{2-13}
         & 50&94&1&0&87&W &60&119&1&0&104&W  \\ \cmidrule{2-13}
         & 54&90&1&0&82&W &64&115&1&0&101&W  \\ \cmidrule{2-13}
         & 58&86&1&0&79&W &68&110&0&1&97&W  \\ \cmidrule{2-13}
         & 62&82&1&0&76&W &72&106&0&1&94&W \\ \bottomrule
    \end{tabular}
\end{table}




\bibliographystyle{model6-num-names}
\bibliography{bili_ref.bib}
\end{document}